\documentclass{article}

\usepackage{graphicx}
\usepackage{amsmath,amssymb,amsthm}
\usepackage{xspace}
\usepackage{fancyhdr}
\usepackage{times}
\usepackage{calc}
\usepackage{gastex}
\usepackage{marvosym}
\usepackage{times}
\usepackage[vlined, linesnumbered]{algorithm2e}

\newcounter{countdownDay}
\newcounter{pagelimit}
\newcommand{\toomanypages}{A\raisebox{-3pt}{a\raisebox{-2pt}{a\raisebox{-1pt}{a}}}rgh~!}
\def\putsmileypage{\ifnum \thepage>\thepagelimit \toomanypages\fi}
\addtocounter{countdownDay}{-\number\day}
\newcount \minutes
\newcount \hour
\newcount \temp
\newcount \hourrem
\newcount \minutesrem
\minutes=\time
\hour=\time
\divide \hour by 60
\temp=\hour
\multiply \temp by -60
\advance \minutes by \temp
\multiply \hour by -1
\multiply \minutes by -1
\advance \hourrem by \hour
\advance \minutesrem by \minutes

\let\emptyset\varnothing
\newcommand{\langeq}{\ensuremath{\approx_L}}

\newcommand{\cC}{\ensuremath{\mathcal{C}}}
\newcommand{\cR}{\ensuremath{\mathcal{R}}}

\newcommand{\N}{\ensuremath{\mathbb{N}}}
\newcommand{\R}{\ensuremath{\mathbb{R}}}  
   \xspace

\newcommand{\tuple}[1]{\ensuremath{\left\langle#1\right\rangle}}

\newcommand{\ecta}{\ensuremath{\sf ECA}\xspace}

\newcommand{\regeq}[1]{\ensuremath{\approx_{#1}}}
\newcommand{\regeqd}[1]{\ensuremath{\approx^\angle_{#1}}}

\newcommand{\constr}[1]{\ensuremath{{\sf Constr}\left(#1\right)}}

\newcommand{\valuations}[1]{\ensuremath{{\cal V}\left(#1\right)}}

\newcommand{\denote}[1]{\ensuremath{\left[\!\left[#1\right]\!\right]}}
\newcommand{\true}{\ensuremath{{\sf true}}}

\newcommand{\cA}{\ensuremath{{\mathcal{A}}}}

\newcommand{\abs}[1]{\ensuremath{\left\lvert #1\right\rvert}}
  \newcommand{\posreal}{\R^{\geq 0}}

  \newcommand{\TF}[1]{\ensuremath{{\sf T}#1^*}}

\newcommand{\Reg}[2]{\ensuremath{{\sf Reg}\left(#1,#2\right)}}
\newcommand{\Regd}[2]{\ensuremath{{\sf Reg}^\angle\left(#1,#2\right)}}

\newcommand{\clocks}[1]{\ensuremath{\mathbb{C}_{#1}}}
\newcommand{\hclocks}[1]{\ensuremath{\mathbb{H}_{#1}}}
\newcommand{\pclocks}[1]{\ensuremath{\mathbb{P}_{#1}}}
\newcommand{\hclock}[1]{\ensuremath{\overleftarrow{x_{#1}}}}
\newcommand{\pclock}[1]{\ensuremath{\overrightarrow{x_{#1}}}}

\newcommand{\cmax}{\ensuremath{\mathit{cmax}}}

\newcommand{\regionautomaton}[3]{\ensuremath{{\sf RegAut}_{#1}^{#2}\left(#3\right)}}

\newcommand{\Arautomaton}[1]{\regionautomaton{\forall}{}{#1}}
\newcommand{\Erautomaton}[1]{\regionautomaton{\exists}{}{#1}}
\newcommand{\ADrautomaton}[1]{\regionautomaton{\forall}{\angle}{#1}}
\newcommand{\EDrautomaton}[1]{\regionautomaton{\exists}{\angle}{#1}}

\newcommand{\utautomaton}[1]{\ensuremath{A^{\sf ut}}}
\newcommand{\release}{\ensuremath{{\sf rel}}}

\newcommand{\untime}[1]{\ensuremath{{\sf Untime}(#1)}}
\newcommand{\TS}[1]{\ensuremath{{\sf TS}_{#1}}}

\newcommand{\roundEC}[1]{\ensuremath{\langle #1\rangle}}
\newcommand{\wplus}{\ensuremath{\operatorname{+_w}}}

\newcommand{\future}[1]{\ensuremath{\overrightarrow{#1}}}
\newcommand{\past}[1]{\ensuremath{\overleftarrow{#1}}}
\newcommand{\postop}{\ensuremath{\mathsf{Post}}\xspace}
\newcommand{\preop}{\ensuremath{\mathsf{Pre}}\xspace}
\newcommand{\post}[1]{\ensuremath{\postop\left(#1\right)}}
\newcommand{\poste}[2]{\ensuremath{\postop_{#1}\left(#2\right)}}
\newcommand{\poststar}[1]{\ensuremath{\postop^*\left(#1\right)}}
\newcommand{\pre}[1]{\ensuremath{\preop\left(#1\right)}}
\newcommand{\pree}[2]{\ensuremath{\preop_{#1}\left(#2\right)}}
\newcommand{\prestar}[1]{\ensuremath{\preop^*\left(#1\right)}}
\newcommand{\posti}[2]{\ensuremath{\postop^{#2}\left(#1\right)}}
\newcommand{\prei}[2]{\ensuremath{\preop^{#2}\left(#1\right)}}

\newcommand{\ezone}{event-zone}
\newcommand{\ezones}{event-zones}
\newcommand{\closure}[1]{\ensuremath{{\rm Closure}_#1}}
\newcommand{\appro}[1]{\ensuremath{{\rm Approx}_#1}}
\newcommand{\plmin}[1]{\ensuremath{#1^{\pm}}}
\newcommand{\forwex}{{\sf ForwExact}\xspace}
\newcommand{\backex}{{\sf BackExact}\xspace}

\SetKwFunction{setc}{SetCst}
\SetKwFunction{dbmnorm}{DBMNormalise}
\SetKwFunction{edbmnorm}{EDBMNormalise}

\renewcommand{\L}{\ensuremath{\mathsf{L}}}
\newcommand{\wL}{\ensuremath{\mathsf{wL}}}

\newtheorem{theorem}{Theorem}
\newtheorem{lemma}{Lemma}
\newtheorem{corollary}{Corollary}
\newtheorem{proposition}{Proposition}
\newtheorem{definition}{Definition}

\newif{\iffalse}
\newif\ifappendix\appendixtrue
\title{Event Clock Automata: from Theory to Practice%
  \thanks{Work supported by the projects: $(i)$ QUASIMODO (FP7-
    ICT-STREP-214755), Quasimodo: ``Quantitative System Properties in
    Model-Driven-Design of Embedded'', {\tt
      http://www.quasimodo.aau.dk/}, $(ii)$ GASICS (ESF-EUROCORES
    LogiCCC), Gasics: ``Games for Analysis and Synthesis of
    Interactive Computational Systems'', {\tt
      http://www.ulb.ac.be/di/gasics/} and $(iii)$ Moves: ``Fundamental
    Issues in Modelling, Verification and Evolution of Software'',
    {\tt http://moves.ulb.ac.be}, a PAI program funded by the Federal
    Belgian Government.}  }

\author{G. Geeraerts${}^1$ \thanks{Partly supported by a `Cr\'edit aux
    chercheurs' from the Belgian FRS/F.N.R.S.}\and J.F. Raskin${}^1$ \and
  N. Sznajder${}^2$ \\ ${}^1$ Universit\'e Libre Bruxelles, D\'epartement
  d'Informatique\\ Brussels, Belgium\\ ${}^2$ Universit\'e Pierre et Marie
  Curie UMR CNRS 7606, LIP6\\ Paris, France.\\ {\tt
    \{gigeerae,jraskin\}@ulb.ac.be}\\ {\tt nathalie.sznajder@lip6.fr}}

\begin{document}
\maketitle

\begin{abstract}
  \emph{Event clock automata} (\ecta) are a model for \emph{timed
    languages} that has been introduced by Alur, Fix and Henzinger as
  an alternative to \emph{timed automata}, with better theoretical
  properties (for instance, \ecta are determinizable while timed
  automata are not).  In this paper, we \emph{revisit} and
  \emph{extend} the theory of \ecta. We first prove that \emph{no
    finite time abstract language equivalence} exists for \ecta,
  thereby disproving a claim in the original work on \ecta. This means
  in particular that regions \emph{do not form a time abstract
    bisimulation}. Nevertheless, we show that regions can still be
  used to build a finite automaton recognizing the \emph{untimed
    language of an \ecta.} Then, we extend the classical notions of
  \emph{zones} and \emph{DBMs} to let them handle event clocks instead
  of plain clocks (as in timed automata) by introducing \emph{event
    zones} and \emph{Event DBMs} (EDBMs). We discuss algorithms to
  handle event zones represented as EDBMs, as well as (semi-)
  algorithms based on EDBMs to decide language emptiness of \ecta.
\end{abstract}

\section{Introduction}
\emph{Timed automata} have been introduced by Alur and Dill in the
early nineties \cite{AD94} and are a successful and popular model to
reason about \emph{timed behaviors} of computer systems. Where finite
automata represent behaviors by finite sequences of actions, timed
automata define sets of \emph{timed words} (called \emph{timed
  languages}) that are finite sequences of actions, each paired with a
real time stamp. To this end, timed automata extend finite automata
with a finite set of real valued clocks, that can be tested and reset
with each action of the system. The theory of timed automata is now
well developed \cite{DBLP:conf/cav/Alur99}. The algorithms to analyse
timed automata have been implemented in several tools such as Kronos
\cite{DBLP:conf/cav/BozgaDMOTY98} or UppAal (which is increasingly
applied in industrial case studies)
\cite{DBLP:conf/qest/BehrmannDLHPYH06}.

Timed automata, however, suffer from certain weaknesses, at least from
the theoretical point of view. As a matter of fact, timed automata are
\emph{not determinizable} and \emph{cannot be complemented} in general
\cite{AD94}.  Intuitively, this stems from the fact that the reset of
the clocks cannot be made deterministic wrt the word being
read. Indeed, from a given location, there can be two transitions,
labeled by the same action $a$ but different reset sets.

This observation has prompted Alur, Fix and Henzinger to introduce the
class of \emph{event clock automata} (\ecta for short) \cite{297329},
as an alternative model for timed languages. Unlike timed automata,
\ecta force the clock resets to be strongly linked to the occurrences
of actions. More precisely, for each action $a$ of the system, there
are two clocks $\hclock{a}$ and $\pclock{a}$ in an \ecta: $\hclock{a}$
is the \emph{history clock} of $a$ and \emph{always records the time
  elapsed since the last occurrence of $a$}. Symmetrically,
$\pclock{a}$ is the \emph{prophecy clock} for $a$, and \emph{always
  predicts the time distance up to the next occurrence of $a$}. As a
consequence, while history clocks see their values \emph{increase}
with time elapsing (like clocks in timed automata do), the values of
prophecy clocks \emph{decrease over time}. However, this scheme
ensures that the value of any clock is uniquely determined at any
point in the timed word being read, no matter what path is being
followed in the \ecta. A nice consequence of this definition is that
\ecta are \emph{determinizable} \cite{297329}. While the theory of
\ecta has witnessed some developments
\cite{339240,Dima:1999:KTE:647899.740957,Tang:2009:EVP:1506688.1506740,GGRS,DBLP:conf/formats/DSouzaT04}
since the seminal paper, no tool is available that exploits the full
power of event clocks (the only tool we are aware of is \textsc{Tempo}
\cite{Sor01} and it is restricted to \emph{event-recording automata},
i.e. \ecta with history clocks only).

In this paper, we revisit and extend the theory of \ecta, with the
hope to make it more practical and amenable to implementation. A
widespread belief \cite{297329} about \ecta and their analysis is that
\ecta are similar enough to timed automata that the classical
techniques (such as regions, zones or DBMs) developed for them can
readily be applied to \ecta. The present research, however, highlights
\emph{fundamental discrepancies} between timed automata and \ecta:
\begin{enumerate}
\item \label{item:2}First, we show that \emph{there is no finite time
    abstract language equivalence} on the valuations of \emph{event
    clocks}, whereas the region equivalence \cite{AD94} \emph{is} a
  finite time abstract language equivalence for timed automata. This
  implies, in particular, that \emph{regions do not form a finite
    time-abstract bisimulation for \ecta}, thereby contradicting a
  claim found in the original paper on \ecta \cite{297329}.
\item With timed automata, checking language emptiness can be done by
  building the so-called region automaton \cite{AD94} which recognizes
  $\untime{L(A)}$, the untimed version of $A$'s timed language. A
  consequence of the surprising result of point~\ref{item:2} is that,
  for some \ecta $A$, the \emph{region automaton recognizes a strict
    subset of $\untime{L(A)}$}. Thus, the region automaton (as defined
  in \cite{AD94}) \emph{is not a sound construction for checking
    language emptiness of \ecta}. We show however that a slight
  modification of the original definition (that we call the
  \emph{existential region automaton}) allows to recover
  $\untime{L(A)}$. Unlike the timed automata case, our proof cannot
  rely on bisimulation arguments, and requires original techniques.
\item Efficient algorithms to analyze timed automata are best
  implemented using \emph{zones} \cite{DBLP:conf/cav/Alur99}, that are
  in turn represented by \emph{DBMs} \cite{Dill89}. Unfortunately,
  zones and DBMs cannot be directly applied to \ecta. Indeed, a zone
  is, roughly speaking, a conjunction of constraints of the form
  $x-y\prec c$, where $x$, $y$ are clocks, $\prec$ is either $<$ or
  $\leq$ and $c$ is an integer. This makes sense in the case of timed
  automata, since the difference of two clock values is an invariant
  with time elapsing. This is not the case when we consider event
  clocks, as \emph{prophecy and history clocks evolve in opposite
    directions with time elapsing}. Thus, we introduce the notions of
  \ezone s and Event DBMs that can handle constraints of the form
  $x+y\prec c$, when $x$ and $y$ are of different types.
\item In the case of timed automata two basic, zone-based algorithms
  for solving language emptiness have been studied: the \emph{forward analysis
    algorithm} that iteratively computes all the states reachable from
  the initial state, and the \emph{backward analysis algorithm} that computes
  all the states that can reach an accepting state. While the former
  might not terminate in general, the latter is guaranteed to terminate
  \cite{DBLP:conf/cav/Alur99}. We show that this is not the case
  anymore with \ecta: \emph{both algorithms might not terminate} again
  because of event clocks evolving in opposite directions.
\end{enumerate}
These observations reflect the structure of the paper. We close it by
discussing the possibility to define \emph{widening operators},
adapted from the \emph{closure by region}, and the
\emph{$k$-approximation} that have been defined for timed automata
\cite{bouyer-fmsd-2004}. The hardest part of this future work will be
to obtain a proof of correctness for these operators, since, here
again, we will not be able to rely on bisimulation arguments.

\section{Preliminaries}
\paragraph{Words and timed words} An alphabet $\Sigma$ is a finite set
of symbols. A (finite) \emph{word} is a finite sequence
$w=w_0w_1\cdots w_n$ of elements of $\Sigma$.  We denote the length of
$w$ by $\abs{w}$.  We denote by $\Sigma^*$ the set of words over
$\Sigma$. A \emph{timed word} over $\Sigma$ is a pair $\theta =
(\tau,w)$ such that $w$ is a word over $\Sigma$ and
$\tau=\tau_0\tau_1\cdots \tau_{\abs{w}-1}$ is a word over $\posreal$
with $\tau_i\leq \tau_{i+1}$ for all $0\leq i < \abs{w}-1$.  We denote
by $\TF{\Sigma}$ the set of timed words over $\Sigma$. A (timed)
\emph{language} is a set of (timed) words. For a timed word
$\theta=(\tau,w)$, we let $\untime{\theta}=w$. For a timed language
$L$, we let $\untime{L}=\{\untime{\theta}\mid\theta\in L\}$.

\paragraph{Event clocks} Given an alphabet $\Sigma$, we define the set
of associated \emph{event clocks}
$\clocks{\Sigma}=\hclocks{\Sigma}\cup\pclocks{\Sigma}$, where
$\hclocks{\Sigma}=\{\hclock{\sigma}\mid \sigma\in\Sigma\}$ is the set
of \emph{history clocks}, and
$\pclocks{\Sigma}=\{\pclock{\sigma}\mid\sigma\in\Sigma\}$ is the set
of \emph{prophecy clocks}. A \emph{valuation} of a set of clocks is a
function $v:C\rightarrow \posreal\cup\{\bot\}$, where $\bot$ means
that the clock value is undefined. We denote by $\valuations{C}$ the
set of all valuations of the clocks in $C$. For a valuation $v\in
\mathcal{V}(C)$, for all $x\in \hclocks{\Sigma}$, we let
$\roundEC{v_1(x)} = \lceil v(x) \rceil - v(x)$ and for all $x\in
\pclocks{\Sigma}$, we let $\roundEC{v(x)} = v(x) - \lfloor v(x)
\rfloor$, where $\lfloor v(x) \rfloor$ and $\lceil v(x) \rceil$ denote
respectively the largest previous and smallest following integer.  We
also denote by $\plmin{v}$ the valuation s.t. $\plmin{v}(x)=v(x)$ for
all $x\in\hclocks{\Sigma}$, and $\plmin{v}(x)=-v(x)$ for all
$x\in\pclocks{\Sigma}$.

For all valuation $v\in\valuations{C}$ and all $d\in\posreal$ such
that $v(x)\geq d$ for all $x\in\pclocks{\Sigma}\cap C$, we define the
valuation $v+d$ obtained from $v$ by letting $d$ time units elapse:
for all $x\in\hclocks{\Sigma}\cap C$, $(v+d)(x)=v(x)+d$ and for all $x\in
\pclocks{\Sigma}\cap C$, $(v+d)(x)=v(x)-d$, with the convention that
$\bot+d=\bot-d=\bot$. A valuation is \emph{initial} iff $v(x)=\bot$
for all $x\in\hclocks{\Sigma}$, and \emph{final} iff $v(x)=\bot$ for
all $x\in\pclocks{\Sigma}$. We note $v[x:=c]$ the valuation that
matches $v$ on all its clocks except for $v(x)$ that equals $c$.

An \emph{atomic clock constraint} over $C\subseteq \clocks{\Sigma}$ is
either \true{} or of the form $x\sim c$, where $x\in C$,
$c\in\mathbb{N}$ and ${\sim}\in\{<,>,=\}$. A \emph{clock constraint}
over $C$ is a Boolean combination of atomic clock constraints. We
denote $\constr{C}$ the set of all possible clock constraints over
$C$. A valuation $v\in\valuations{C}$ satisfies a clock constraint
$\psi\in\constr{C}$, denoted $v\models\psi$ according to the following
rules: $v\models\true$, $v\models x\sim c$ iff $v(x)\sim c$,
$v\models\neg\psi$ iff $v\not\models\psi$, and
$v\models\psi_1\wedge\psi_2$ iff $v\models\psi_1$ and
$v\models\psi_2$.

\paragraph{Event-clock automata} An \emph{event-clock automaton} $A=
\tuple{Q,q_i,\Sigma,\delta,\alpha}$ (\ecta for short) is a tuple,
where $Q$ is a finite set of locations, $q_i\in Q$ is the initial
location, $\Sigma$ is an alphabet, $\delta\subseteq
Q\times\Sigma\times\constr{\clocks{\Sigma}}\times Q$ is a finite set
of edges, and $\alpha\subseteq Q$ is the set of accepting
locations. We additionally require that, for each $q\in Q$,
$\sigma\in\Sigma$, $\delta$ is defined for a finite number of
$\psi\in\constr{\clocks{\Sigma}}$.  An \emph{extended state} (or
simply state) of an \ecta $A=\tuple{Q,q_i,\Sigma,\delta,\alpha}$ is a
pair $(q,v)$ where $q\in Q$ is a location, and
$v\in\valuations{\clocks{\Sigma}}$ is a valuation.

\paragraph{Runs and accepted language}
The semantics of an \ecta $A=\tuple{Q,q_i,\Sigma,\delta,\alpha}$ is
best described by an infinite transition system $\TS{A}=\tuple{Q^A,
  Q_i^A, \rightarrow, \alpha^A}$, where
$Q^A=Q\times\valuations{\clocks{\Sigma}}$ is the set of extended
states of $A$, $Q_i^A=\{(q_i,v)\mid v \textrm{ is initial}\}$,
$\alpha^A=\{(q,v)\mid q\in \alpha\textrm{ and } v \textrm{ is
  final}\}$. The transition relation ${\rightarrow}\subseteq
Q^A\times\posreal\times Q^A\cup Q^A\times\Sigma\times Q^A$ is
s.t. $(i)$ $\big((q,v),t, (q,v')\big){\in}\rightarrow$ iff $v'=v+t$
(we denote this by $(q,v)\xrightarrow{t}(q,v')$), and $(ii)$
$\big((q,v),\sigma,(q',v')\bigr)\in\rightarrow$ iff there is
$(q,\sigma,\psi,q')\in \delta$ and
$\overline{v}\in\valuations{\clocks{\Sigma}}$
s.t. $\overline{v}[\pclock{\sigma}:=0]=v$,
$\overline{v}[\hclock{\sigma}:=0]=v'$ and $\overline{v}\models\psi$
(we denote this $(q,v)\xrightarrow{\sigma}(q',v')$).  We note
$(q,v)\xrightarrow{t,\sigma}(q',v')$ whenever there is $(q'',v'')$
s.t. $(q,v)\xrightarrow{t} (q'',v'')\xrightarrow{\sigma}
(q',v'')$. Intuitively, this means that an history clock
$\hclock{\sigma}$ always records the time elapsed since the last
occurrence of the corresponding $\sigma$ event, and that a prophecy
clock $\pclock{\sigma}$ always predicts the delay up to the next
occurrence of $\sigma$. Thus, when firing a $\sigma$-labeled
transition, the guard must be tested against $\overline{v}$ (as
defined above) because it correctly predicts the next occurrence of
$\sigma$ and correctly records its last occurrence (unlike $v$ and
$v'$, as $v(\pclock{\sigma})=0$ and $v'(\hclock{\sigma})=0$)

A sequence
$(q_0,v_0)(t_0,w_0)(q_1,v_1)(t_1,w_1)(q_2,v_2)\cdots(q_n,v_n)$ is a
$(q,v)$-\emph{run} of $A$ on the timed word $\theta=(\tau,w)$ iff:
$(q_0,v_0)=(q,v)$, $t_0=\tau_0$, for any $1\leq i\leq n-1$:
$t_i=\tau_i-\tau_{i-1}$, and for any $0\leq i\leq n-1$:
$(q_i,v_i)\xrightarrow{t_i,w_i}(q_{i+1},v_{i+1})$. A $(q,v)$-run is
\emph{initialized} iff $(q,v)\in Q_i^A$ (in this case, we simply call
it a run). A $(q,v)$-run on $\theta$, ending in $(q_n,v_n)$ is
\emph{accepting} iff $(q_n,v_n)\in\alpha^A$. In this case, we say that
the run accepts $\theta$. For an \ecta $A$ and an extended state
$(q,v)$ of $A$, we denote by $L(A,(q,v))$ the set of timed words
accepted by a $(q,v)$-run of $A$, and by $L(A)$ the set of timed words
accepted by an initialized run of $A$.

\section{Equivalence relations for event-clocks\label{sec:equiv-relat-event}}
A classical technique to analyze timed transition systems is to define
\emph{time abstract equivalence relations} on the set of states, and
to reason on the \emph{quotient} transition system. In the case of
\emph{timed automata}, a fundamental concept is the \emph{region
  equivalence} \cite{AD94}, which is a \emph{finite time-abstract}
bisimulation, and allows to decide properties of timed automata such
as reachability. Contrary to a widespread belief \cite{297329}, we
show that the class of \ecta does not benefit of these properties, as
\textbf{\ecta admit no finite time-abstract language equivalence}.

\paragraph{Time-abstract equivalence relations}
Let $\cC$ be a class of \ecta, all sharing the same alphabet
$\Sigma$. We recall three equivalence notions on event clock
valuations:
\begin{itemize}
\item
  ${\lesssim}\subseteq\valuations{\clocks{\Sigma}}\times\valuations{\clocks{\Sigma}}$
  is a \emph{time abstract simulation relation} for the class $\cC$
  iff, for all $\cA\in\cC$, for all location $q$ of $\cA$, for all
  $(v_1, v_2)\in {\lesssim}$, for all $t_1\in\posreal$, for all
  $a\in\Sigma$: $(q,v_1)\xrightarrow{t_1,a}(q',v_1')$ implies that
  there exists $t_2\in\posreal$
  s.t. $(q,v_2)\xrightarrow{t_2,a}(q',v_2')$ and $v_1'\lesssim
  v_2'$. In this case, we say that \emph{$v_2$ simulates
    $v_1$}. Finally,
  ${\simeq}\subseteq\valuations{\clocks{\Sigma}}\times\valuations{\clocks{\Sigma}}$
  is a \emph{time abstract simulation equivalence} iff there exists a
  time abstract simulation relation $\lesssim$
  s.t. ${\simeq}=\{(v_1,v_2)\mid v_1\lesssim v_2\textrm{ and }
  v_2\lesssim v_1\}$

\item $\sim$ is a \emph{time abstract bisimulation equivalence} for
  the class $\cC$ iff it is a \emph{symmetric} time abstract
  simulation for the class $\cC$.
\item
  ${\langeq}\subseteq\valuations{\clocks{\Sigma}}\times\valuations{\clocks{\Sigma}}$
  is a \emph{time abstract language equivalence} for the class $\cC$
  iff for all $\cA\in\cC$, for all location $q$ of $\cA$, for all
  $(v_1, v_2)\in{\langeq}$: $\untime{L(q, v_1)}=\untime{L(q, v_2)}$
\end{itemize}
We say that an equivalence relation is \emph{finite} iff it is of
finite index. Clearly, any time abstract bisimulation is a time
abstract simulation equivalence, and any time abstract simulation
equivalence is a time abstract language equivalence. We prove the
absence of \emph{finite} time abstract language equivalence for \ecta,
thanks to $A_{\sf inf}$ depicted in
Fig.~\ref{fig:no-finite-bisimulation}:
\begin{proposition}\label{prop:no-finite-ta-le}
  There is no finite time abstract language equivalence for \ecta.
\end{proposition}
\begin{proof} 
  Let us assume that $\langeq$ is a time abstract language equivalence
  on the class of \ecta. We will show, thanks to $A_{\sf inf}$, that
  $\langeq$ has necessarily \emph{infinitely} many equivalence
  classes.
  
  For any $n\in \N$, let $v^n$ denote the \emph{initial} valuation of
  $\clocks{\{a,b\}}$ s.t.  $v^n(\pclock{a})=n$ and
  $v^n(\pclock{b})=0$, and let $\theta^n$ denote the timed word
  $(b,0)(b,1)(b,2)\cdots(b,n-1)(a,n)$. Observe that, for any $n\geq 0$,
  there is only one run of $A_{\sf inf}$ starting in $(q_0,v^n)$ and
  this run accepts $\theta^n$. Hence, for any $n\geq 0$: $\untime{\L(A,
    (q_0,v^n))}=\untime{\{\theta^n\}}=a^nb$. 
  
  Now, let $j,k$ be two natural values with $j\neq k$. Let $s^j=(q_0,
  v^j)$ and $s^k=(q_0,v^k)$. Clearly, $v^j\not\langeq v^k$ since
  $\untime{\L(A, s^j)}\neq \untime{\L(A_{\sf inf}, s^k)}$. Since this
  is true for infinitely many pairs $(v^j, v^k)$, $\langeq$ has
  necessarily an infinite number of equivalence classes. Thus, there
  is no \emph{finite} time abstract language equivalence on the class
  of \ecta. 
\end{proof}

\begin{corollary}\label{cor:no-finite-ecta}
  There is no \emph{finite} time abstract language equivalence, no
  \emph{finite} time abstract simulation equivalence and no
  \emph{finite} time abstract bisimulation for \ecta.
\end{corollary}

\begin{figure}[t]
  \begin{center}
    \begin{tabular}{lclc}
    \includegraphics[scale=1]{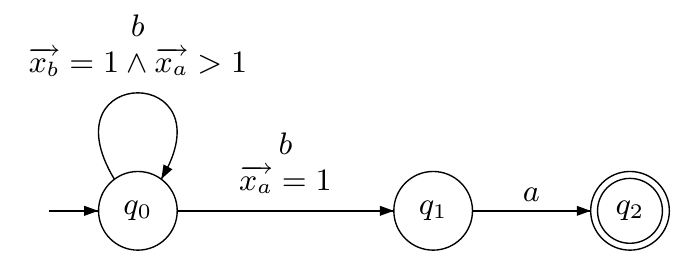}
  \end{tabular}
  \end{center}

  \caption{The automaton $A_{\sf
      inf}$}\label{fig:no-finite-bisimulation}
\end{figure}

\section{Regions and event clocks\label{sec:regions-event-clocks}}
For the class of timed automata, the \emph{region equivalence} has been
shown to be a \emph{finite time-abstract bisimulation}, which is used
to build the so-called \emph{region automaton}, a finite-state
automaton recognizing $\untime{L(A)}$ for all timed automata $A$
\cite{AD94}. Corollary~\ref{cor:no-finite-ecta} tells us that regions
are not a time-abstract bisimulation for \ecta (contrary to what was
claimed in \cite{297329}). Let us show that we can nevertheless rely
on the notion of region to build a finite automaton recognizing
$\untime{L(A)}$ for all \ecta $A$.

\paragraph{Regions} Let us fix a set of clocks
$C\subseteq\clocks{\Sigma}$ and a constant $\cmax\in\mathbb{N}$. We
first recall two region equivalences from the literature. The former,
denoted $\regeq{\cmax}$, is the classical Alur-Dill region equivalence
for timed automata \cite{AD94} while the latter (denoted
$\regeqd{\cmax}$) is adapted from Bouyer \cite{bouyer-fmsd-2004} and
refines the former:
\begin{itemize}
\item For any $v_1,v_2\in \valuations{C}$: $v_1\regeq{\cmax} v_2$ iff:
\begin{itemize}
\item[$({\sf C1})$] for all $x\in C$, $v_1(x)=\bot$ iff $v_2(x)=\bot$,
\item[$({\sf C2})$] for all $x\in C$: either $v_1(x) >\cmax$ and
  $v_2(x)>\cmax$, or $\lceil v_1(x)\rceil = \lceil v_2(x)\rceil$ and
  $\lfloor v_1(x)\rfloor = \lfloor v_2(x) \rfloor$,
\item[$({\sf C3})$] for all $x_1$, $x_2\in C$ s.t. $v_1(x_1)\leq\cmax$
  and $v_1(x_2)\leq\cmax$: $\roundEC{v_1(x_1)} \leq
  \roundEC{v_1(x_2)}$ if and only if $\roundEC{v_2(x_1)} \leq
  \roundEC{v_2(x_2)}$.
\end{itemize}
\item For all $v_1,v_2\in \valuations{C}$: $v_1\regeqd{\cmax} v_2$
  iff: $v_1\regeq{\cmax} v_2$ and:
\begin{itemize}
\item[$({\sf C4})$] For all $x_1,x_2\in C$ s.t. $v_1(x_1)> \cmax$ or
  $v_1(x_2)> \cmax$: either we have
  $\abs{\plmin{v_1}(x_1)-\plmin{v_1}(x_2)}>2\cdot\cmax$ and
  $\abs{\plmin{v_2}(x_1)-\plmin{v_2}(x_2)}>2\cdot\cmax$; or we have
  $\lfloor \plmin{v_1}(x_1)-\plmin{v_1}(x_2)\rfloor=\lfloor
  \plmin{v_2}(x_1)-\plmin{v_2}(x_2)\rfloor$ and $\lceil
  \plmin{v_1}(x_1)-\plmin{v_1}(x_2)\rceil=\lceil
  \plmin{v_2}(x_1)-\plmin{v_2}(x_2)\rceil$.
\end{itemize}
\end{itemize}

Equivalence classes of both $\regeq{\cmax}$ and $\regeqd{\cmax}$ are
called \emph{regions}. We denote by $\Reg{C}{\cmax}$ and
$\Regd{C}{\cmax}$ the set of regions of $\regeq{\cmax}$ and
$\regeqd{\cmax}$ respectively. Fig.~\ref{fig:diffreg} $(a)$, $(b)$ and
$(c)$ illustrate these two notions. Comparing $(a)$ and $(b)$ clearly
shows how $\regeqd{\cmax}$ refines $\regeq{\cmax}$ by introducing
diagonal constraints between clocks larger than $\cmax$. Moreover,
$(c)$ shows why we need to rely on $\plmin{v_1}$ and $\plmin{v_2}$ in
${\sf C4}$: in this case, $C$ contains an history and a prophecy clock
that evolve in opposite directions with time elapsing. Thus, their
\emph{sum} remains constant over time (hence the $2\cdot\cmax$ in~${\sf C4}$).

\begin{figure}
  \begin{center}
    \begin{tabular}{ll}
      (a)\\
      &\includegraphics[scale=.5]{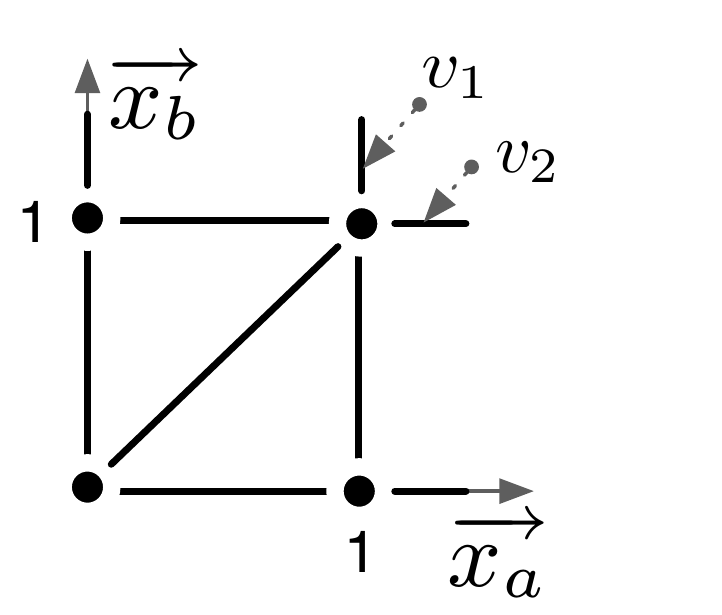}
    \end{tabular}

    \begin{tabular}{ll}
      (b)\\
      & \includegraphics[scale=.5]{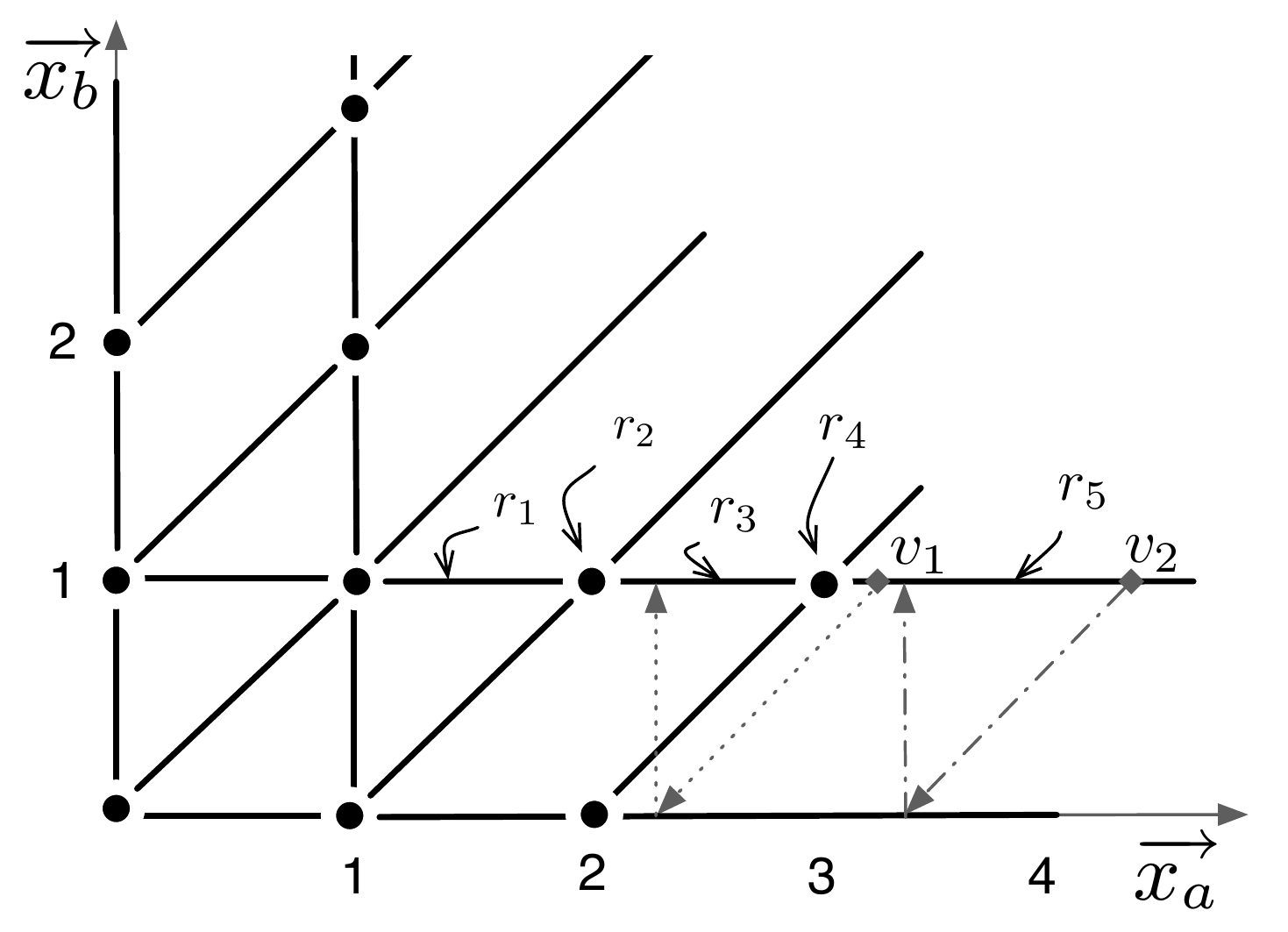}
    \end{tabular}

  \begin{tabular}{ll}
      (c)\\
      & \includegraphics[scale=.5]{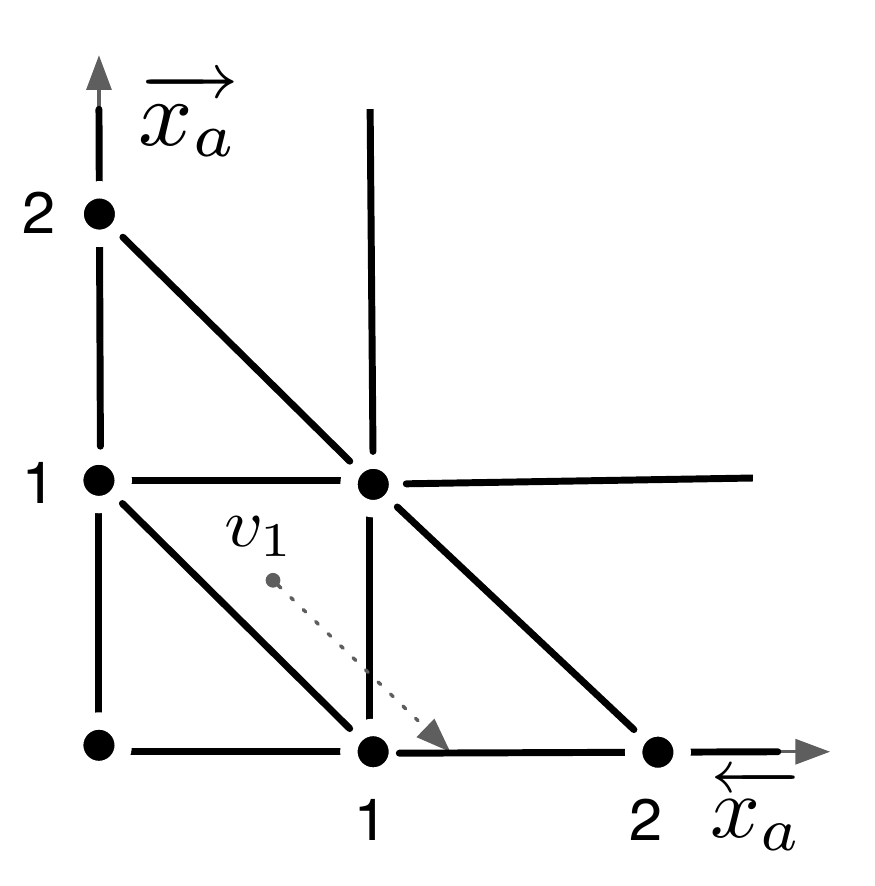}
    \end{tabular}
  \end{center}
  \caption{The sets of regions $(a)$ $\Reg{\pclocks{\{a,b\}}}{1}$,
    $(b)$ $\Regd{\pclocks{\{a,b\}}}{1}$ and $(c)$
    $\Regd{\clocks{\{a\}}}{1}$. Dotted arrows show the trajectories
    followed by the valuations with time elapsing. Curved arrows are
    used to refer to selected regions.}
 \label{fig:diffreg}
\end{figure}

Observe that, for any $\cmax$, and for any finite set of clocks $C$,
$\Reg{C}{\cmax}$ and $\Regd{C}{\cmax}$ are \emph{finite} sets. A
region $r$ on set of clocks $C$ is \emph{initial} (resp. \emph{final})
iff it contains only initial (final) valuations.

\paragraph{Regions are not a language equivalence} Since both notions
of regions defined above are finite,
Corollary~\ref{cor:no-finite-ecta} implies that they cannot form a
language equivalence for \ecta. Let us explain intuitively why it is
not the case. Consider $\Reg{\pclocks{\{a,b\}}}{1}$ and the two
valuations $v_1$ and $v_2$ in Fig.~\ref{fig:diffreg} (a). Clearly,
$v_1$ can reach the region where $\pclock{a}=1$ and $\pclock{b}>1$,
while $v_2$ cannot. Conversely, $v_2$ can reach $\pclock{a}>1$ and
$\pclock{b}=1$ but $v_2$ cannot. It is easy to build an \ecta with
$\cmax=1$ that distinguishes between those two cases and accepts
different words. Then, consider $\Regd{\pclocks{\{a,b\}}}{1}$ and the
valuations $v^3$ and $v^4$ (not shown in the figure)
s.t. $v^3(\pclock{b})=v^4(\pclock{b})=1$, $v^3(\pclock{a})=4$ and
$v^4(\pclock{a})=5$. It is easy to see that for $A_{\sf inf}$ in
Fig.~\ref{fig:no-finite-bisimulation}:
$\untime{L(A_{\inf},(q_0,v^3))}=\{bbba\}\neq\{bbbba\}=\untime{L(A_{\sf
    inf},(q_0,v^4))}$, although $v^3$ and $v^4$ belong to the same
region. Indeed, from $v^3$, the $(q_0,q_0)$ loop can be fired 3 times
before we reach $\pclock{a}=1$ and the $(q_0,q_1)$ edge can be
fired. However, the $(q_0,q_0)$ loop has to be fired 4 times from
$v^4$ before we reach $\pclock{a}=1$ and the $(q_0,q_1)$ edge can be
fired. Remark that these are essentially the same arguments as in the
proof of Proposition~\ref{prop:no-finite-ta-le}.  These two examples
illustrate the issue with \emph{prophecy clocks} and regions. Roughly
speaking, to keep the set of regions finite, valuations where the
clocks are \emph{too large} (for instance, $>\cmax$ in the case of
$\Reg{C}{\cmax}$) belong to the same region. This is not a problem for
history clocks as an history clock larger than $\cmax$ remains over
$\cmax$ with time elapsing. This is not the case for prophecy clocks
whose values \emph{decrease with time elapsing}: eventually, those
clocks reach a value $\leq\cmax$, but the region equivalence is too
coarse to allow to predict the region they reach.

\paragraph{Region automata} Let us now consider the consequence of
Corollary~\ref{cor:no-finite-ecta} on the notion of region automaton.
We first define two variants of the region automaton:
\begin{definition}\label{strregauto}
  Let $A= \tuple{Q,q_i,\Sigma,\delta,\alpha}$ and $\cR$ be a set of
  regions on $\valuations{\clocks{\Sigma}}$. Then, the
  \textbf{existential} (resp. \textbf{universal}) \emph{$\cR$-region
    automaton of $\cA$} is the finite transition system
  $RA(\exists,\cR,A)$ (resp. $RA(\forall, \cR, A)$) defined by 
  $\tuple{Q^R, Q_i^R, \Sigma, \delta^R, \alpha^R}$ s.t.:
  \begin{enumerate}
  \item $Q^R= Q\times \cR$
  \item $Q_i^R= \{(q_i,r)\mid r \textrm{ is an initial region}\}$
  \item $\delta^R\subseteq Q^R\times\Sigma\times Q^R$ is
    s.t. $\big((q_1,r_1), a, (q_2, r_2)\big)\in \delta$ iff
    \textbf{there exists a valuation} (resp. \textbf{for all
      valuations}) $v_1\in r_1$, there exists a time delay
    $t\in\posreal$ and a valuation $v_2\in r_2$
    s.t. $(q_1,v_1)\xrightarrow{t,a}(q_2,v_2)$.
  \item $\alpha^R=\{(q,r)\mid q\in \alpha\textrm{ and $r$ is a final
      region}\}$
  \end{enumerate}
\end{definition}
Let $R=\tuple{Q^R, Q_i^R, \Sigma, \delta^R, \alpha^R}$ be a region
automaton and $w$ be an (untimed) word over $\Sigma$. A \emph{run} of
$R$ on $w=w_0w_1\ldots w_n$ is a finite sequence
$(q_0,r_0)(q_1,r_1)\ldots$ $(q_{n+1},r_{n+1})$ of states of $R$ such that:
$(q_0,r_0)\in Q_i^R$ and such that: for all $0\leq i\leq n$:
$\big((q_i,r_i),w_{i}, (q_{i+1},r_{i+1})\big)\in\delta^R$. Such a
run is \emph{accepting} iff $(q_{n+1},r_{n+1})\in\alpha^R$ (in that case, we say
that $w$ is accepted by $R$). The language $L(R)$ of $R$ is the set of
all untimed words accepted by $R$.

Let $A$ be an \ecta with alphabet $\Sigma$ and maximal constant
$\cmax$. If we adapt and apply the notion of region automaton, as
defined for timed automata \cite{AD94}, to $A$ we obtain
$RA(\forall,\Reg{\clocks{\Sigma}}{\cmax},A)$. To alleviate notations,
we denote it by $\Arautomaton{A}$. In the rest of the paper, we also
consider three other variants: $(i)$~$\ADrautomaton{A}=$ $RA(\forall,
\Regd{\clocks{\Sigma}}{\cmax},A)$, $(ii)$~$\Erautomaton{A}=$
$RA(\exists, \Reg{\clocks{\Sigma}}{\cmax},A)$ and
$(iii)$~$\EDrautomaton{A}=RA(\exists,
\Regd{\clocks{\Sigma}}{\cmax},A)$. Observe that, for timed automata,
all these automata coincide, and thus accept the untimed language
(this can be proved by a bisimulation argument) \cite{AD94}. Let us
see how these results adapt (or not) to \ecta.

\paragraph{Recognized language of universal region automata} Let us
show that, in general \emph{universal} region automata \emph{do not
  recognize the untimed language of the \ecta.}
\begin{lemma}\label{lem:univ-do-not-recognise}
  There is an \ecta $A$ such that
  $L(\Arautomaton{A})\subsetneq\untime{L(A)}$ and such that
  $L(\ADrautomaton{A})\subsetneq\untime{L(A)}$.
\end{lemma}
\begin{proof}
  Consider the automaton $A_{\sf inf}$ in
  Fig.~\ref{fig:no-finite-bisimulation}, with $\cmax=1$. Assume there
  is, in $\Arautomaton{A_{\sf inf}}$, and edge of the form
  $\big((q_0,r),b,(q_0,r')\big)$, where $r$ is initial. By the guard
  of the $(q_0,q_0)$ loop, $r'$ is a region s.t. for all $v\in r$:
  $v(\pclock{b})=1$ and $v(\pclock{a})>1$. To fire the $(q_0,q_0)$
  loop again, we need to let time elapse up to the point where
  $\pclock{b}=0$. Then consider two valuations $v$ and $v'$
  s.t. $v(\pclock{b})=v'(\pclock{b})=1$, $v(\pclock{a})=1.1$ and
  $v'(\pclock{a})=2.1$. Clearly, $\{v,v'\}\subseteq r'$. However,
  firing the $(q_0,q_0)$ loop from $(q_0,v)$ leads to $(q_0, v'')$,
  with $v''(\pclock{a})=0.1$, and firing the same $(q_0,q_0)$ loop
  from $(q_0,v')$ leads to $(q_0,v''')$ with
  $v'''(\pclock{a})=1.1$. Thus, $v''$ and $v'''$ do not belong to the
  same region. Since we are considering a \emph{universal} automaton,
  we conclude that there is no edge of the form
  $\big((q_0,r'),b,(q_0,r'')\big)$. Hence, $\Arautomaton{A_{\sf inf}}$
  cannot recognize an arbitrary number of $b$'s from any of its
  initial states, and thus, $L(\Arautomaton{A_{\sf
      inf}})\subsetneq\untime{L(A_{\sf inf})}$.

  For the second case, we consider Fig.~\ref{fig:diffreg} (b) that
  depicts the projection of the set of regions used to build
  $\ADrautomaton{A_{\sf inf}}$ on the clocks
  $\{\pclock{a},\pclock{b}\}$ (remark that we can restrict our
  reasoning to this projection, since the other clocks are never
  tested in $A_{\sf inf}$). Assume there is, in $\ADrautomaton{A_{\sf
      inf}}$, an edge of the form $\big((q_0,r),b,(q_0,r')\big)$ were
  $r$ is initial. This implies that $r'\in\{r_1,\ldots, r_5\}$ (we
  refer to the names in Fig.~\ref{fig:diffreg}), because of the guard
  of the $(q_0,q_0)$ loop. Since $\untime{L(A_{\sf inf})}=\{b^na\mid
  n\geq 1\}$, it must be possible to accept an arbitrary number of
  $b$'s from one of the $(q_0,r')$. Let us show that it is not the
  case.  From $r_3$ and $r_4$ we have edges
  $\big((q_0,r_3),b,(q_0,r_1)\big)$ and
  $\big(q_0,r_4),b,(q_0,r_2)\big)$. However, there is no valuation
  $v\in r_1\cup r_2$ s.t. $(v+t)(\pclock{b})=0$ and
  $(v+t)(\pclock{a})>1$ for some $t$. Thus, there is, in
  $\ADrautomaton{A_{\sf inf}}$, no edge of the form
  $\big((q_0,r),b,(q_0,r')\big)$ when $r\in r_1, r_2$. Finally, there
  is no edge of the form $\big((q_0, r_5),b,(q_0,r)\big)$ because
  \emph{some valuations} of $r_5$ (such as $v_1$) will reach $r_3$ and
  \emph{some others} (such as $v_2$) will stay in $r_5$ after the
  firing of the loop. Since we consider a \emph{universal} automaton,
  $(q_0, r_5)$ has no successor.
\end{proof}

\paragraph{Recognized language of existential region automata}
Fortunately, the definition of \emph{existential region automaton}
allows us to recover a finite transition system recognizing exactly
$\untime{L(A)}$, for all \ecta $A$.  Remark that our construction is
\emph{direct}, contrary to the original construction \cite{297329}
that consists in first translating the \ecta into a non-deterministic
timed automaton recognising the same timed language but with an
increased number of clocks compared to the original \ecta, and then
computing the region automaton of this timed automaton. Moreover, the
proof we are about to present cannot invoke the fact that regions form
a time-abstract bisimulation, as it is the case for timed automata,
and we thus need to rely on different proof techniques. Actually, we
will show that:
\begin{eqnarray*}
  \untime{L(A)}\subseteq 
  L(\EDrautomaton{A})\subseteq 
  L(\Erautomaton{A})\subseteq 
  \untime{L(A)}
\end{eqnarray*}
The two leftmost inequalities are easily established by the following
reasonings. Let $(q_0,v_0)(t_0,w_0)(q_1,v_1)$
$(t_1,w_1)\cdots(q_n,v_n)$ be an accepting run of $A$ on
$\theta=(\tau,w)$. Thus, $\theta\in L(A)$. For all $0\leq i\leq n$ let
$r_i$ be the (unique) region containing $v_i$. Then, by definition of
$\EDrautomaton{A}$, $(q_0,r_0)w_0(q_1,r_1)w_1\cdots(q_n,r_n)$ is an
accepting run of $\EDrautomaton{A}$ on $w=\untime{\theta}$. Hence
$\untime{L(A)}\subseteq L(\EDrautomaton{A})$. Second, since
$\regeqd{\cmax}$ refines $\regeq{\cmax}$, each accepting run
$(q_0,r_0)w_0(q_1,r_1)w_1\cdots(q_n,r_n)$ in $\EDrautomaton{A}$
corresponds to an accepting run
$(q_0,r_0')w_0(q_1,r_1')w_1\cdots(q_n,r_n')$ in $\Erautomaton{A}$,
where for any $0\leq i\leq n$, $r_i'$ is the (unique) region of
$\Reg{\clocks{\Sigma}}{\cmax}$ that contains $r_i$. Hence,
$L(\EDrautomaton{A})\subseteq L(\Erautomaton{A})$.

To establish $L(\Erautomaton{A})\subseteq \untime{L(A)}$ we need to
rely on the notion of \emph{weak time successor}. The set of
\emph{weak time successors} of $v$ by $t$ time units is:
{\small
\begin{eqnarray*}
  v\wplus t&=&
  \left\{
    \begin{array}{r|cc}
      &&\big(x\in\pclocks{\Sigma}\textrm{ and }v(x)>\cmax\big)\textrm{ implies } v'(x)>\cmax-t\\
      v'&\forall x:&\textrm{and}\\
      &&\big(x\notin\pclocks{\Sigma}\textrm{ or }v(x)\leq\cmax\textrm{ or } v(x)=\bot\big)\textrm{ implies }v'(x)=(v+t)(x)
    \end{array}
  \right\}
\end{eqnarray*}
} As can be seen, weak time successors introduce non-determinism on
prophecy clocks that are larger than $\cmax$. So, $v\wplus t$ is a
\emph{set} of valuations.  Let $q$ be a location of an \ecta. We write
$(q,v)\xrightarrow{t}_w(q,v')$ whenever $v'\in(v\wplus t)$. Then, a
sequence $(q_0,v_0)(t_0,w_0)$
$(q_1,v_1)(t_1,w_1)(q_2,v_2)\cdots(q_n,v_n)$ is an initialized
\emph{weak run}, on $\theta=(\tau,w)$, of an \ecta $A=\tuple{Q, q_i,
  \Sigma, \delta, \alpha}$ iff $q_0= q_i$, $v_0$ is initial,
$t_0=\tau_0$, for any $1\leq i\leq n-1$: $t_i=\tau_i-\tau_{i-1}$, and
for any $0\leq i\leq n-1$: there is $(q_i',v_i')$ s.t.
$(q_i,v_i)\xrightarrow{t_i}_w(q_i',v_i')\xrightarrow{w_i}(q_{i+1},v_{i+1})$.
A weak run is accepting iff $q_n\in\alpha$ and $v_n$ is final. The
weak language $\wL(A)$ of $A$ is the set of all timed words $\theta$
s.t. there is an accepting weak run on $\theta$. Clearly,
$L(A)\subseteq \wL(A)$ as every run is also a weak run. However, the
converse also holds, since the non-determinism appears only on clocks
larger than $\cmax$, which the automaton cannot distinguish:
\begin{proposition}\label{prop:wl-equals-l}
  For any \ecta $A$: $L(A)=\wL(A)$.
\end{proposition}
\begin{proof}
    Since, by definition, every run is a weak run, $L(A)\subseteq
  \wL(A)$. Let us show that $L(A)\supseteq \wL(A)$. Let
  $\theta=(\tau_0,w_0)\cdots(\tau_n,w_n)$ be a timed word in $\wL(A)$,
  and let
  $(q_0,v_0)\xrightarrow{t_0}_w(q_0,v_0')\xrightarrow{w_0}(q_1,v_1)
  \cdots(q_{n+1},v_{n+1})$ be the corresponding accepting weak run of
  $A$. For any $0\leq i\leq n$, we build $\overline{v}_i$ as
  follows. For any $x$ s.t. $x\in\pclocks{\Sigma}$ and
  $v'_i(x)>\cmax$, let $k>i$ be the least position
  s.t. $v'_k(x)\leq\cmax$. Remark that such a position always exists
  in an \emph{accepting} run (recall that if a letter is never to be
  seen again, its valuation must be set to $\bot$). Then, we let
  $\overline{v}_i(x)=v'_k(x)+\sum_{j=i+1}^kt_j$. Otherwise, we let
  $\overline{v}_i(x)=v_i'(x)$. Remark that $v_i'$ and $\overline{v}_i$
  differ only on prophecy clocks larger than $\cmax$, and that
  $v_i'(x)>cmax$ iff $\overline{v}_i(x)>\cmax$ for any $i$ and
  $x$. Moreover, the definition of the sequence of $\overline{v}_i$
  clearly respects the definition of time successor. We further define
  $\tilde{v}_i$ for all $i$ as follows: $\tilde{v}_i(x)=t_i +
  \overline{v}_i(x)$ for all $x\in\pclocks{\Sigma}$
  s.t. $v_i(x)>\cmax$ and $\tilde{v}_i(x) = v_i(x)$ otherwise. Hence,
  it can be checked that for all $0\leq i\leq n$,
  $(q_i,\tilde{v}_i)\xrightarrow{t_i}(q_{i},
  \overline{v}_{i})\xrightarrow{w_i}(q_{i+1},\tilde{v}_{i+1})$, and so
  that
  $(q_0,\tilde{v}_0)\xrightarrow{t_0}(q_0,\overline{v}_0)\xrightarrow{w_0}(q_1,\tilde{v}_1)
  \xrightarrow{t_1}(q_1,\overline{v}_1)\cdots(q_{n+1},\tilde{v}_{n+1})$. Moreover,
  $\tilde{v}_{n+1}(x)=v_{n+1}(x)=\bot$ for all
  $x\in\pclocks{\Sigma}$. Thus, $\theta\in L(A)$ and thus,
  $L(A)\supseteq \wL(A)$.
 
\end{proof}

Then, we prove that \emph{weak time successors} enjoy a property which
is reminiscent of time abstract bisimulation. This allows to establish
Theorem~\ref{th:correctness-regaut}.
\begin{lemma}\label{lemma:wts}
  Let $C$ be a set of clocks and let $\cmax$ be a natural
  constant. For any $v_1,v_2\in\valuations{C}$ s.t. $v_1\regeq{\cmax}
  v_2$, for any $t_1\in\posreal$, there exist $t_2$ and $v'\in
  (v_2\wplus t_2)$ s.t. $v_1+t_1\regeq{\cmax}v'$.
\end{lemma}

\begin{proof}
    The cases where $v_1\regeq{\cmax}v_1+t_1$ are trivial. We first
  restrict ourselves to the case where $v_1$ and $v_1+t_1$ belong
  to adjacent regions, that is:
  \begin{eqnarray}
    \exists 0<t\leq t_1:
    \left(
      \begin{array}{c}
        \forall 0\leq t'\leq t: v_1+t'\regeq{\cmax}v_1\\
        \textrm{and}\\
        \forall t<t'\leq t_1: v_1+t'\regeq{\cmax}v_1+t_1
      \end{array}\right)\label{eq:7}
  \end{eqnarray} 
  Let us now show how to chose $t_2$. Let $C^0_{v}$ denote the set
  of clocks $x$ s.t. $\roundEC{v(x)}=0$. Under the hypothesis
  (\ref{eq:7}), we have to consider two cases:
  \begin{enumerate}
  \item Either $C^0_{v_1}=\emptyset$ and
    $C^0_{v_1+t_1}\neq\emptyset$. In that case, let $x$ be a clock
    in $C^0_{v_1+t_1}$. We let $t_2=\roundEC{v_2(x)}$
  \item Or $C^0_{v_1}\neq\emptyset$ and
    $C^0_{v_1+t_1}=\emptyset$. In that case, we need to consider
    two sub-cases. If there is $x$ s.t. $\roundEC{v_2(x)}\neq 0$, we
    let $t_2$ be a value
    s.t. $0<t_2<\min\{\roundEC{v_2(x)}\mid\roundEC{v_2(x)}\neq
    0\}$. Otherwise, all the clocks in $v_2$ have a null fractional
    part, and we can take any delay $<1$ for $t_2$: we let
    $t_2=0.1$.
  \end{enumerate}
  
  Now, let us show that there exists $v\in v_2\wplus t_2$
  s.t. $v\regeq{\cmax}v_1+t_1$. For that purpose, we first build a
  valuation $v_3$ as follows. For any history clock $x$, we let
  $v_3(x)=v_2(x)$. For all prophecy clocks $x$ s.t. $v_2(x)\leq\cmax$, or $v_2(x)=\bot$,
  we let $v_3(x)=v_2(x)$ too. For all prophecy clocks $x$
  s.t. $v_2(x)>\cmax$ (and thus $v_1(x)>\cmax$ since
  $v_1\regeq{\cmax}v_2$), we consider two cases. Either
  $(v_1+t_1)(x)>\cmax$. In that case we let
  $v_3(x)=\cmax+t_2+1$. Or $(v_1+t_1)(x)=\cmax$. In that case we
  let $v_3(x)=\cmax+t_2$. Remark that the case
  $(v_1+t_1)(x)<\cmax$ is not possible since we have assumed that
  $v_1(x)>\cmax$ and that $v_1$ and $v_1+t_1$ are in adjacent
  regions.

  We now let $v'= v_3+t_2$. It is easy to check that
  $v'\regeq{\cmax}(v_1+t_1)$. Moreover, $v'\in(v_2\wplus t_2)$,
  since $v_3$ has been obtained from $v_2$ by replacing values larger
  than $\cmax$ by other values larger than $\cmax$.

  To conclude, observe that if $v_3\in(v_2\wplus t_2)$ and
  $v_2\in(v_1\wplus t_1)$, then
  $v_3\in(v_1\wplus(t_1+t_2))$. This allows to handle the case
  where $v_1$ and $v_1+t_1$ are not in adjacent regions: by
  decomposing $t_1$ into a sequence
  $t_1',t_2',\ldots,t_n'$
  s.t. $t_1=t_1'+t_2'+\cdots+t_n'$, and for all $1\leq
  i<n$, $v_1+\sum_{j=1}^it'_j$ and $v_1+\sum_{j=1}^{i+1}t'_{j}$
  are in adjacent regions. Then, applying the reasoning above, we get
  a sequence $t_1'',\ldots, t_n''$ of time delays and a sequence
  $v_0',v_1',\ldots,v_n'$ of valuations s.t. $v_0'=v_2$, for all
  $0\leq i<n$, $v_{i+1}'\in v_i'\wplus t_i''$ and
  $v'_{i+1}\regeq{\cmax}v_1+\sum_{j=1}^{i+1}t_j'$. Thus, $v_n'\in
  v_2\wplus\sum_{j=1}^nt_j''$ and
  $v_n'\regeq{\cmax}v_1+\sum_{j=1}^{n}t_j'=v_1+t_1$.
\end{proof}

We can now prove that:
\begin{theorem}\label{th:correctness-regaut}
  For any \ecta $A=(\Sigma,Q,q_i,\delta,\alpha)$:
  $\L(\Erautomaton{A})\subseteq\untime{\L(A)}$. 
\end{theorem}
\begin{proof}
  Let $(q_0,r_0)\xrightarrow{w_0}(q_1,r_1)\xrightarrow{w_1}\cdots
  \xrightarrow{w_{n-1}}(q_n,r_n)$ be an accepting run of
  $\Erautomaton{A}$. Let us build, inductively a sequence
  $\overline{t}_0$, $\overline{t}_1$,\ldots, $\overline{t}_{n-1}$ of
  time delays and a sequence $\overline{v}_0$,
  $\overline{v}_1$,\ldots, $\overline{v}_n$ of valuations
  s.t. $\overline{v}_i\in r_i$ for all $0\leq i\leq n$. This will
  allow us to obtain an accepting weak run of $A$. For the base case,
  we let $\overline{v}_0$ be a valuation from $r_0$ and we let
  $\overline{v}_1$ and $\overline{t}_0$ be
  s.t. $\overline{v}_0\xrightarrow{\overline{t}_0,w_0}\overline{v}_1$
  with $\overline{v}_1\in r_1$. Such $\overline{v}_1$ and
  $\overline{t}_0$ are guaranteed to exist by definition of the region
  automaton, and since $(q_0,r_0)\xrightarrow{w_0}(q_1,r_1)$ in this
  region automaton. For the inductive case, we consider $i$ with
  $2\leq i\leq n$ and assume that $\overline{v}_{i-1}$ has been
  defined and is in $r_{i-1}$. Let us show how to build
  $\overline{t}_{i-1}$ and $\overline{v}_i$. Since
  $r_{i-1}\xrightarrow{w_{i-1}}r_i$ in the region automaton, there are
  $v_i\in r_i$, $v_{i-1}\in r_{i-1}$, and $t_{i-1}$
  s.t. $v_{i-1}\xrightarrow{t_{i-1}}v_{i-1}+t_{i-1}\xrightarrow{w_{i-1}}v_i$. Let
  $c$ denote the value $v_i(\pclock{w_{i-1}})$.  Since
  \begin{eqnarray}
    v_{i-1}+t_{i-1}\xrightarrow{w_{i-1}}v_i\label{eq:2}
  \end{eqnarray}
  we know that 
  \begin{eqnarray}
    (v_{i-1}+t_{i-1})[\pclock{w_{i-1}}:=c]\models\psi\label{eq:3}
  \end{eqnarray}
  where $\psi$ is
  the guard of the edge responsible for
  $v_{i-1}+t_{i-1}\xrightarrow{w_{i-1}}v_i$ and that
  \begin{eqnarray}
    v_i=(v_{i-1}+t_{i-1})[\pclock{w_{i-1}}:=c,\hclock{w_{i-1}}:=0].\label{eq:4}
  \end{eqnarray}

  Next, we let $v_{i-1}'$ be a valuation and $\overline{t}_{i-1}$ be a
  time delay s.t. $v_{i-1}'\in\overline{v}_{i-1}\wplus t_{i-1}$ and
  \begin{eqnarray}
    v_{i-1}'\regeq{\cmax}(v_{i-1}+\overline{t}_{i-1}).\label{eq:5}
  \end{eqnarray}
  Such $v_{i-1}'$ and
  $\overline{t}_{i-1}$ are guaranteed to exist by
  Lemma~\ref{lemma:wts}: $\overline{v}_{i-1}\in r_{i-1}$ by induction
  hypothesis and $v_{i-1}\in r_{i-1}$ by construction, hence
  $\overline{v}_{i-1}\regeq{\cmax}\overline{v}_{i-1}$. Then, we let
  \begin{eqnarray}
    \overline{v}_i=v_{i-1}'[\pclock{w_{i-1}}:=c,\hclock{w_{i-1}}:=0]\label{eq:6}
  \end{eqnarray}

  Let us check that $v_{i-1}'\xrightarrow{w_i}\overline{v}_i$. By
  (\ref{eq:5}), $v_{i-1}'$ and $v_{i-1}+t_{i-1}$ are
  equivalent. Hence,
  $v_{i-1}'[\pclock{w_{i-1}}:=c]\regeq{\cmax}(v_{i-1}+t_{i-1})[\pclock{w_{i-1}}:=c]$. Thus,
  by (\ref{eq:3}), $v_{i-1}'[\pclock{w_{i-1}}:=c]\models\psi$, and the
  same transition can be fired from $v_{i-1}'$, leading to
  $\overline{v}_i$, by~(\ref{eq:6}). Finally,
  by~(\ref{eq:4}),~(\ref{eq:6}) and~(\ref{eq:5}), we deduce that
  $\overline{v}_i\regeq{\cmax} v_i\in r_i$, hence $\overline{v}_i\in
  r_i$.
 
  By construction,
  $(q_0,\overline{v}_0)\xrightarrow{\overline{t}_0,w_0}(q_1,\overline{v}_1)
  \xrightarrow{\overline{t}_1,w_1}\cdots(q_n,\overline{v}_n)$ is an
  accepting weak run of $A$ on $\theta$ with
  $\untime{\theta}=w$. Thus,
  $\wL(\Erautomaton{A})\subseteq\untime{\L(A)}$. Since
  $\wL(\Erautomaton{A})=L(\Erautomaton{A})$, by
  Proposition~\ref{prop:wl-equals-l}, we have
  $L(\Erautomaton{A})\subseteq\untime{L(A)}$.
\end{proof}

\paragraph{Size of the existential region automaton} The number of
Alur-Dill regions on $n$ clocks and with maximal constant $\cmax$ is
at most $R(n,\cmax)=n!\times 2^n\times (2\times\cmax+2)^n$
\cite{AD94}. Adapting this result to take into account the $\bot$
value, we have: $|\Reg{\clocks{\Sigma}}{\cmax}|\leq R(2\times
|\Sigma|, \cmax+1)$.  Hence, the number of locations of
$\Erautomaton{A}$ for an $\ecta$ $A$ with $m$ locations and alphabet
$\Sigma$ is at most $m\times R(2\times|\Sigma|,\cmax+1)$.  In
\cite{297329}, a technique is given to obtain a finite automaton
recognizing $\untime{L(A)}$ for all \ecta~$A$: first transform $A$
into a non-deterministic timed automaton \cite{AD94} $A'$
s.t. $L(A')=L(A)$, then compute the region automaton of $A'$. However,
building $A'$ incurs a blow up in the number of clocks and locations,
and the size of the region automaton of $A'$ is at most $m\times
2^K\times R(K,\cmax)$ where $K=6\times|\Sigma|\times(\cmax+2)$ is an
upper bound on the number of atomic clock constraints in $A$. Our
construction thus yields a smaller automaton.

\section{Zones and event-clocks\label{sec:zones-event-clocks}}
In the setting of timed automata, the \emph{zone datastructure}
\cite{Dill89} has been introduced as an effective way to improve the
running time and memory consumption of on-the-fly algorithms for
checking emptiness. In this section, we \emph{adapt} this notion to
the framework of \ecta, and discuss forward and backward analysis
algorithms.  Roughly speaking, a \emph{zone} is a symbolic
representation for a set of clock valuations that are defined by
constraints of the form $x-y\prec c$, where $x,y$ are clocks, $\prec$
is either $<$ or $\leq$, and $c$ is an integer constant. Keeping the
difference between clock values makes sense in the setting of timed
automata as all the clocks have always real values and the difference
between two clock values is an invariant over the elapsing of time. To
adapt the notion of zone to \ecta, we need to overcome two
difficulties. First, prophecy and history clocks evolve in different
directions with time elapsing. Hence, it is not always the case that
if $v(x)-v(y)=c$ then $(v+t)(x)-(v+t)(y)=c$ for all $t$ (for instance
if $x$ is a prophecy clocks and $y$ an history clock). However, the
\emph{sum} of clocks of different types is now an invariant, so event
clock zones must be definable, either by constraints of the form
$x-y\prec c$, if $x$ and $y$ are both history or both prophecy clocks,
or by constraints of the form $x+y\prec c$ otherwise. Second, clocks
can now take the special value $\bot$.  Formally, we introduce the
notion of \ezone{} as follows.

\begin{definition}\label{def:zone}
  For a set $C$ of clocks over an alphabet $\Sigma$, an \emph{\ezone}
  is a subset of $\valuations{C}$ that is defined by a conjunction of
  constraints of the form $x = \bot$; $x \sim c$; $x_1-x_2 \sim c$ if
  $x_1,x_2\in\hclocks{\Sigma}$ or $x_1,x_2\in \pclocks{\Sigma}$; and
  $x_1+x_2\sim c$ if either $x_1\in\hclocks{\Sigma}$ and
  $x_2\in\pclocks{\Sigma}$ or $x_1\in \pclocks{\Sigma}$ and
  $x_2\in\hclocks{\Sigma}$, with $x,x_1,x_2\in C$,
  ${\sim}\in\{\leq,\geq, <, >\}$ and $c\in \mathbb{Z}$.
\end{definition}

\paragraph{Event-clock Difference Bound Matrices}
In the context of timed automata, Difference Bound Matrices (DBMs for
short) have been introduced to represent and manipulate zones
\cite{Bellman57,Dill89}. Let us now adapt DBMs to event clocks.  In
order to adapt DBMs to \ezone s, we need to be able to $(i)$ encode
contraints of the form $x+y\prec c$ and of the form $x'-y'\prec c$,
depending on the types of $x$, $y$, $x'$ and $y'$, $(ii)$ encode
constraints of the form $x=\bot$, and $(iii)$ encode the fact that a
variable is not constrained by the zone. Indeed, in a DBM, this is
encoded by the pair of constraints $x\geq 0$ and $x<+\infty$. This is
not sound in our case since $0\leq x<+\infty$ implies that $x\neq
\bot$. Thus, we introduce a special symbol $?$ to denote the absence
of constraint.

Formally, an EDBM $M$ of the set of clocks $C=\{x_1,\ldots, x_n\}$ is
a $(n+1)$ square matrix of elements from $\bigl (\mathbb{Z}\times
\{<,\leq\}\bigr )\cup \{(\infty,<), (\bot,=),(?,=)\}$ s.t. for all
$0\leq i,j,\leq n$: $m_{i,j}=(\bot,=)$ implies $i=0$ or $j=0$ (i.e.,
$\bot$ can only appear in the first position of a row or
column). Thus, a constraint of the form $x_i=\bot$ will be encoded
with either $m_{i,0}=(\bot,=)$ or $m_{0,i}=(\bot,=)$. As in the case
of DBMs, we assume that the extra clock $x_0$ is always equal to
zero. Moreover, since prophecy clocks decrease with time evolving,
they are encoded by their \emph{opposite value} in the matrix. Hence
the EDBM naturally encodes \emph{sums} of variables when the two
clocks are of different types. Each element $(m_{ij},\prec_{ij})$ of
the matrix thus represents either the constraint $x_i-x_j\prec_{ij}
m_{ij}$ or the constraint $x_i+x_j\prec_{ij} m_{ij}$, depending on the
type of $x_i$ and $x_j$.  Finally, the special symbol $?$ encodes the
fact that the variable is not constrained (it can take any real value,
or the $\bot$ value). Formally, an EDBM $M$ on set of clocks
$C=\{x_1,\ldots,x_n\}$ represents the zone $\denote{M}$ on set of
clocks $C$ s.t. $v\in\denote{M}$ iff for all $0\leq i,j\leq n$:
\textbf{if} $M_{i,j}=(c,\prec)$ with $c\neq\ ?$
\textbf{then} $\plmin{v}(x_i)-\plmin{v}(x_j)\prec c$ (assuming
$\plmin{v}(x_0)$ denotes the value $0$ and assuming that for all
$k\in\mathbb{Z}\cup\{\bot\}$:
$\bot+k=\bot-k=k+\bot=k-\bot=\bot$). When $\denote{M}=\emptyset$, we
say that $M$ is \emph{empty}. In the sequel, we also rely on the
$\leq$ ordering on EDBM elements. We let $(m;\prec)\leq (m';\prec')$
iff one of the following holds: either $(i)$ $m'={?}$; or $(ii)$
$m,m'\in\mathbb{Z}\cup\{\infty\}$ and $m<m'$; or $(iii)$ $m=m'$ and
either $\prec=\prec'$ or $\prec'=\leq$.

As an example, consider the two following EDBMs that both represent
$x_1=\bot \wedge 0 < x_3 - x_4 < 1 \wedge x_2 + x_4 \leq 2$ (where
$x_1,x_2$ are prophecy clocks, and $x_3, x_4$ are history clocks): 
$$
\begin{pmatrix}
(0,\leq) & (\bot,=) & (?,=) & (?,=) & (?,=)\\
(\bot,=) & (?,=) & (?,=) & (?,=) & (?,=)\\
(0,\leq) & (?,=) & (0, \leq) & (?,=) & (?,=)\\
(?,=) & (?,=) & (?,=) & (0,\leq) & (1, <)\\
(?,=) & (?,=) & (2,\leq) & (0,<) & (0,\leq)\\
\end{pmatrix}
$$

$$
\begin{pmatrix}
(0,\leq) & (\bot,=) & (\infty,<) & (0,\leq) & (0,\leq)\\
(\bot,=) & (?,=) & (?,=) & (?,=) & (?,=)\\
(0,\leq) & (?,=) & (0, \leq) & (0,\leq) & (0,\leq)\\
(\infty,<) & (?,=) & (\infty,<) & (0,\leq) & (1, <)\\
(\infty,<) & (?,=) & (2,\leq) & (0,<) & (0,\leq)\\
\end{pmatrix}
$$

\paragraph{Normal form EDBMs} As in the case of DBMs, we define a
\emph{normal form} for EDBM, and show how to turn any EDBM $M$ into a
normal form EDBM $M'$ s.t. $\denote{M}=\denote{M'}$.  A non-empty EDBM
$M$ is in \emph{normal form} iff the following holds: $(i)$ for all
$1\leq i\leq n$: $M_{i,0}=(\bot,=)$ iff $M_{0,i}=(\bot,=)$ and
$M_{i,0}=(?,=)$ iff $M_{0,i}=(?,=)$, $(ii)$ for all $1\leq i\leq n$:
$M_{i,0}\in\{(\bot,=),(?,=)\}$ implies $M_{i,j}=M_{j,i}=(?,=)$ for all
$1\leq j\leq n$, $(iii)$ for all $1\leq i,j \leq n: M_{i,j}=(?,=)$ iff
either $M_{i,0}\in\{(?,=),(\bot,=)\}$ or
$M_{j,0}\in\{(?,=),(\bot,=)\}$ and $(iv)$ the matrix $M'$ is a
\emph{normal form DBM} \cite{Dill89}, where $M'$ is obtained by
projecting away all lines $1\leq i\leq n$
s.t. $M_{i,0}\in\{(?,=),(\bot,=)\}$ and all columns $1\leq j\leq n$
s.t. $M_{0,j}\in\{(?,=),(\bot,=)\}$ from $M$. To canonically represent
the empty zone, we select a particular EDBM $M_\emptyset$
s.t. $\denote{M_\emptyset}=\emptyset$. For example, the latter EDBM of
the above example is in normal form.

Then, given an EDBM $M$, Algorithm~\ref{algo:normEDBM} allows to
compute a normal form EDBM $M'$ s.t. $\denote{M}=\denote{M'}$. This
algorithm relies on the function \dbmnorm{$M$,$S$}\!, where $M$ is an
$(\ell+1)\times(\ell+1)$ EDBM, and $S\subseteq\{0,\ldots,\ell\}$.
\dbmnorm{$M$,$S$} applies the classical normalisation algorithm for
DBMs \cite{Dill89} on the DBM obtained by projecting away from $M$ all
the lines and columns $i\not\in S$. Algorithm~\ref{algo:normEDBM}
proceeds in three steps. In the first loop, we look for lines
(resp. columns) $i$ s.t. $M_{i,0}$ (resp. $M_{0,i}$) is $(\bot,=)$,
meaning that there is a constraint imposing that $x_i=\bot$. In this
case, the corresponding $M_{0,i}$ (resp. $M_{i,0}$) must be equal to
$(\bot,=)$ too, and all the other elements in the $i$th line and
column must contain $(?,=)$. If we find a $j$ s.t.  $M_{i,j}\neq
(?,=)$ or $M_{j,i}\neq (?,=)$, then the zone is empty, and we return
$M_\emptyset$. Then, in the second loop, the algorithm looks for lines
(resp. columns) $i$ with the first element equal to $(?,=)$ but
containing a constraint of the form $(c,\prec)$, which imposes that
the variable $i$ must be different from~$\bot$. We record this
information by replacing the $(?,=)$ in $M_{i,0}$ (resp. $M_{0,i}$) by
the weakest possible constraint that forces $x_i$ to have a value
different from~$\bot$. This is either $(0,\leq)$ or $(\infty, <)$,
depending on the type of $x_i$ and is taken care by the \setc{}
function. At this point the set $S$ contains the indices of all
variables that are constrained to be real. The algorithm finishes by
calling the normalisation function for DBMs. Remark, in particular,
that the algorithm returns $M_\emptyset$ iff $M$ is empty which also
provides us with a test for EDBM emptiness.

\begin{proposition}
  For all EDBM $M$, \edbmnorm{$M$} returns a normal form EDBM $M'$
  s.t. $\denote{M'}=\denote{M}$.
\end{proposition}

\begin{algorithm}[t]
  \edbmnorm{$M$}
  \Begin{
    Let $S = \{0\}$ \;
    \ForEach{$1\leq i\leq n$ s.t. $M_{i,0}=(\bot,=)$ or $M_{0,i}=(\bot,=)$}{
      \lIf{$\exists 1\leq j\leq n$ s.t. $M_{i,j}\neq (?,=)$ or $M_{j,i}\neq(?,=)$} {\Return{$M_\emptyset$\;}}
      $M_{i,0} \leftarrow (\bot,=)$ ; $M_{0,i}\leftarrow (\bot,=)$ \;
    }

    \ForEach{$0\leq i,j\leq n$ s.t. $M_{i,j}\notin\{ (?,=),(\bot,=)\}$}{
      $S\leftarrow S\cup\{i,j\}$ \;
         
    }
    \lForEach{$i,j\in S$}{
      \setc{$M_{i,j}$} \;
    }
    $M'\leftarrow$ \dbmnorm{$M$,$S$} \;
    \lIf{$M'=\texttt{Empty}$}{\Return{$M_\emptyset$} \;}
    
    \Return{$M'$} \;
  }
  \medskip
  
  \setc{$M_{i,j}$}
  \Begin {
    \If{$M_{i,j} = (?,=)$} {
      \lIf{$x_i\in\pclocks{\Sigma}$ and ($x_j\in\hclocks{\Sigma}$ or $x_j=x_0$)}{$M_{i,j}\leftarrow (0,\leq)$ \;}
      \lElse{$M_{i,j}\leftarrow (\infty,<)$ \;}
    }
  }
  \caption{A normalisation algorithm for EDBMs.\label{algo:normEDBM}}
\end{algorithm}

\paragraph{Operations on zones}
The four basic operations we need to perform on \ezone s are: $(i)$
\emph{future} of an \ezone{} $Z$ : $\future{Z}=\{v\in
\mathcal{V}(\clocks{\Sigma}) \mid \exists v'\in Z, t\in\posreal:
v=v'+t\}$; $(ii)$~\emph{past} of an \ezone{} $Z$ :
$\overleftarrow{Z}=\{v\in\mathcal{V}(\clocks{\Sigma})\mid \exists
t\in\posreal: v+t\in Z\}$; $(iii)$~\emph{intersection} of two
\ezones{} $Z$ and $Z'$; and $(iv)$ \emph{release} of a clock $x$ in
$Z$: $\release_x(Z) = \{v[x:=d]\mid v\in Z, d\in\posreal\cup\{\bot\}\}$.
Moreover, we also need to be able to test for inclusion of two zones
encoded as EDBMs.  Let $M$, $M_1$ and $M_2$ be EDBMs in normal form,
on $n$ clocks. Then:
\begin{description}
\item[Future] If $M=M_\emptyset$, we let
  $\future{M}=M_\emptyset$. Otherwise, we let $\future{M}$ be s.t.:
\begin{eqnarray*}
  \future{M}_{i,j}&=&\begin{cases}
    (0,\leq) & \textrm{if $M_{ij}\notin\{(\bot,=),(?,=)\}$, $j=0$ and $x_i\in\pclocks{\Sigma}$}\\
    (\infty, <) & \textrm{if $M_{ij}\notin\{(\bot,=),(?,=)\}$, $j=0$ and $x_i\in\hclocks{\Sigma}$}\\
    M_{i,j} & \textrm{otherwise}\\
  \end{cases}
\end{eqnarray*}
\item[Past] If $M=M_\emptyset$, we let
  $\past{M}=M_\emptyset$. Otherwise, we let $\past{M}$ be s.t. for all
  $i$, $j$:
\begin{eqnarray*}
  \past{M}_{i,j}&=&\begin{cases}
    (\infty,<) & \textrm{if $M_{ij}\notin\{(\bot,=),(?,=)\}$, $i=0$ and $x_j\in \pclocks{\Sigma}$}\\
    (0,\leq) & \textrm{if $M_{ij}\notin\{(\bot,=),(?,=)\}$, $i=0$ and $x_j\in \hclocks{\Sigma}$}\\
    M_{i,j} & \textrm{otherwise}
  \end{cases}
\end{eqnarray*}
\item[Intersection] We consider several cases. If $M^1=M_\emptyset$ or
  $M^2=M_\emptyset$, we let $M^1\cap M^2=M_\emptyset$. If there are
  $0\leq i,j\leq n$ s.t. $M^1_{i,j}\not\leq M^2_{i,j}$ and
  $M^2_{i,j}\not\leq M^1_{i,j}$, we let $M^1\cap M^2=M_\emptyset$
  too. Otherwise, we let $M^1\cap M^2$ be the EDBM $M'$ s.t for all
  $i,j$: $M'_{i,j}=min(M^1_{i,j},M^2_{i,j})$.
\item[Release] Let $x$ be an event clock. In the case where
  $M=M_\emptyset$, we let $\release_x(M)=M_\emptyset$. Otherwise, we
  let $\release_x(M)$ be the EDBM s.t. for all $i,j$:
  \begin{eqnarray*}
    \release_x(M)_{i,j}&=&
    \begin{cases}
      M_{i,j}&\textrm{if }x_i\neq x\textrm{ and } x_j\neq x\\
      (?,=)&\textrm{otherwise}
    \end{cases}
  \end{eqnarray*}
\item[Inclusion] We note $M^1\subseteq M^2$ iff $M^1_{i,j}\leq
  M^2_{i,j}$ for all $0\leq i,j\leq n$.
\end{description}

\begin{proposition}\label{prop:construction-EDBMs}\label{prop:inclusion-EDBMs}
  Let $M,M^1,M^2$ be EDBMs in normal form, on set of clocks $C$. Then,
  $(i)$ $\future{\denote{M}}=\denote{\future{M}}$, $(ii)$
  $\past{\denote{M}}=\denote{\past{M}}$, $(iii)$ $\denote{M^1\cap M^2}
  = \denote{M^1}\cap\denote{M^2}$, $(iv)$ for all clock $x\in C$,
  $\release_x(\denote{M})=\denote{\release_x(M)}$ and $(v)$
  $\denote{M^1}\subseteq\denote{M^2}$ iff $M^1\subseteq M^2$.
\end{proposition}
\begin{proof}
  \begin{enumerate}
\item In the case where $M=M_\emptyset$ the proof is trivial. Otherwise,
  $M$ is non-empty, since it is in normal form. We
  assume that $M$ is an EDBM on set of clocks $C=\{x_1,\ldots, x_n\}$,
  that for all $0\leq i,j\leq n$: $M_{i,j}=(m_{i,j},\prec_{i,j})$ and
  that $\future{M}=(m'_{i,j},\prec'_{i,j})$.  It is easy to see that
  any $v\in\future{\denote{M}}$ satisfies the constraints of
  $\denote{\future{M}}$. Thus,
  $\future{\denote{M}}\subseteq\denote{\future{M}}$.

 Consider now a valuation $v\in
  \denote{\future{M}}$.  We need to find a delay $t\in\posreal$ such
  that there exists $v_M\in \denote{M}$ such that $v_M +t = v$. This
  amounts to solving the following system of inequalities:
{\small 
  \begin{displaymath}
    \begin{cases}
      -m_{i0} - v(x_i) \prec_{i0} t \prec_{0i} m_{0i} - v(x_i) & \textrm{for all $x_i\in\pclocks{\Sigma}\cap C$ such that $m_{0i}\notin\{\bot,?\}$}\\
      v(x_i) - m_{i0}\prec_{i0} t \prec_{0i} v(x_i) + m_{0i} & \textrm{for all $x_i\in\hclocks{\Sigma}\cap C$ such that $m_{0i}\notin\{\bot,?\}$} \\
      0\leq t
    \end{cases}
  \end{displaymath}
}
  with the convention that $\infty+c=\infty-c=\infty$ and that
  $-\infty+c=-\infty-c=-\infty$ for all $c\in\mathbb{N}$.  We show that
  the set of solutions is not empty, i.e. that all inequalities are
  pairwise coherent.

  Since for all $x_i\in\pclocks{\Sigma}\cap C$, $(m'_{0i},\prec'_{0i})
  = (m_{0i},\prec_{0i})$, we know that $v(x_i)\prec_{0i} m_{0i}$ and
  since for all $x_i\in\hclocks{\Sigma}\cap C$ $(m'_{0i},\prec'_{0i})
  = (m_{0i},\prec_{0i})$ , we also know that $-m_{0i}\prec_{0i}
  v(x_i)$. Then, none of the inequalities forces $t$ to be negative.

  Let now $x_i,x_j$ be two prophecy clocks
  s.t. $m_{0,i}\not\in\{\bot,?\}$ and $m_{0,j}\not\in\{\bot,?\}$. For
  all $v_M\in\denote{M}$, $-m_{i0}\prec_{i0} v_M(x_i) \prec_{0i}
  m_{0i}$, and $-m_{j0}\prec_{j0} v_M(x_j) \prec_{0j} m_{0j}$, then
  $-m_{i0} - m_{0j} \prec_1 v_M(x_i)-v_M(x_j)\prec_2 m_{0i}+m_{j0}$,
  where $\prec_1 = \leq$ iff $\prec_{i0} = \leq$ and $\prec_{0j} =
  \leq$ $\prec_2=\leq$ iff $\prec_{0i} = \leq$ and $\prec_{j0} =
  \leq$.  Since $M$ is in normal form, $(m_{ji}, \prec_{ji}) \leq
  (m_{0i} + m_{j0}, \prec_2)$ and $(m_{ij},\prec_{ij})\leq (m_{i0} +
  m_{0j},\prec_1)$. Since $(m'_{ij},\prec'_{ij})=(m_{ij},\prec_{ij})$
  and $(m'_{ji},\prec'_{ji})=(m_{ji},\prec_{ji})$, we deduce that
  $-m_{i0} - m_{0j} \prec_1 v(x_i)-v(x_j)\prec_2
  m_{0i}+m_{j0}$. Hence, $-m_{i0}-v(x_i) \prec_1 m_{0j} - v(x_j)$ and
  $-m_{j0} - v(x_j) \prec_2 m_{0i} - v(x_i)$. Then the constraints on
  $t$ deduced from $x_i$ and $x_j$ are coherent. With the same
  arguments, we obtain that the constraints on $t$ deduced from
  $x_i,x_j\in\hclocks{\Sigma}\cap C$ are coherent too.
 
  Consider now $x_i\in\pclocks{\Sigma}\cap C$ and
  $x_j\in\hclocks{\Sigma}\cap C$. Then again, since any valuation
  ${v}_M$ in $\denote{M}$ satisfies $-m_{i0}-m_{0j} \prec_1 {v}_M(x_i)
  + {v}_M(x_j) \prec_2 m_{0i} + m_{j0}$, so does $v$, and one can
  deduce that $-m_{i0} - v(x_i) \prec_1 v(x_j)+m_{0j}$ and
  $v(x_j)-m_{j0} \prec_2 m_{0i} - v(x_i)$ and hence that the
  constraints on $t$ derived from $x_i\in\pclocks{\Sigma}\cap C$ and
  $x_j\in\hclocks{\Sigma}\cap C$ are coherent.
 
  Then, the set of solutions of the inequalities is not empty. Let $t$
  be such a solution. We let $v_M$ be the valuation s.t. $v_M(x) =
  v(x)+ t$ for any $x\in\pclocks{\Sigma}\cap C$ and $v_M(x)=v(x)-t$
  for all $x\in \hclocks{\Sigma}\cap C$. Such a valuation exists, and
  is in $\denote{M}$ by construction.  Then, since $v = v_M+t$ with
  $v_M\in \denote{M}$ and some $t\in\posreal$ we deduce that $v\in
  \future{\denote{M}}$. We conclude that
  $\denote{\future{M}}\subseteq\future{\denote{M}}$.
\item As prophecy and history clocks evolve in opposite directions, the
  arguments of the proof for $\future{M}$ can be adapted.
  \item   In the case where $M^1=M_\emptyset$ or $M^2=M_\emptyset$ the proof
  is trivial. Otherwise, $M^1$ and $M^2$ are non-empty, since they are
  in normal form.  First consider the case where there are $0\leq
  i,j\leq n$ s.t. $M^1_{i,j}\not\leq M^2_{i,j}$ and $M^2_{i,j}\not\leq
  M^1_{i,j}$. By definition of $\leq$, this implies that either
  $M^1_{i,j}$ or $M^2_{i,j}$ is equal to $(\bot,=)$, and that the
  other constraint is of the form $(\prec,m)$, with
  $m\in\posreal\cup\{\infty\}$. Then, clearly
  $\denote{M^1}\cap\denote{M^2}=\emptyset$ and thus
  $\denote{M^1}\cap\denote{M^2}=\denote{M_\emptyset}=\denote{M^1\cap
    M^2}$.

  Thus, let us assume that for all $0\leq i,j\leq n$,
  $\min\{M^1_{i,j},M^2_{i,j}\}$ is defined.  Let $v$ be a valuations
  on the set of clocks $C=\{x_1,\ldots,x_n\}$, let $M$ be an EDBM on
  $C$. Then for all $0\leq i,j\leq n$, we say that $v$ satisfies
  $M_{i,j}=(m_{i,j},\prec_{i,j})$ (denoted $v\models
  M_{i,j}=(m_{i,j},\prec_{i,j})$) iff:
  \begin{enumerate}
  \item either $m_{i,j}=?$
  \item or $i=0$ and $m_{i,j}=v(x_j)=\bot$
  \item or $j=0$ and $m_{i,j}=v(x_i)=\bot$
  \item or $m_{i,j}\not\in\{?,\bot\}$ and
    $|x_i|-|x_j|\prec_{i,j}m_{i,j}$, assuming
    $\bot+c=c+\bot=\bot-c=c-\bot=\bot$ for all $c$.
  \end{enumerate}
  Then, clearly, $\denote{M}=\{v\mid \forall 0\leq i,j\leq n: v\models
  M_{i,j}\}$.
  
  Then observe that, by definition of the ordering $\leq$ on EDBM
  constraints:
  \begin{eqnarray*}
    \big(v\models (m_1,\prec_1)\textrm{ and } v\models (m_2,\prec_2)\big)
    &\textbf{ iff }&
    v\models\min\big\{(m_1,\prec_1),(m_2,\prec_2)\big\}
  \end{eqnarray*}
  By definition of $M^1\cap M^2$, we conclude that
  $\denote{M^1}\cap\denote{M^2}=\denote{M^1\cap M^2}$.
  \item In the case where $M=M_\emptyset$ the proof is trivial. Otherwise,
  $M$ is non-empty, since it is in normal form. Let us assume that $x$
  is the clock of index $k$ in $C$. We first examine the case where
  $M_{k0}=(?,=)$, then $\release_x(M)=M$ since $M$ is in normal form.
  Since $x$ is already unconstrained in $M$, we have
  $\release_x(\denote{M})=\denote{M}$. Hence
  $\release_x(\denote{M})=\denote{M}=\denote{\release_x(M)}$.

  Otherwise, let us assume that $C=\{x_1,\ldots,x_n\}$ and that for
  all $0\leq i,j\leq n$: $M_{ij}=(m_{ij},\prec_{ij})$.  Let $v\in
  \release_x(\denote{M})$. Then there is some $v'\in \denote{M}$, such
  that $v'(y)=v(y)$, for all clock $y\neq x$ in C. Since $v'$
  satisfies all the constraints of $M$, $v$ satisfies all the
  constraints of $\denote{M}$ related to clocks different from $x$,
  and hence $v\in \denote{\release_x(M)}$. Thus,
  $\release_x(\denote{M})\subseteq \denote{\release_x(M)}$.

  Conversely, let $v\in\denote{\release_x(M)}$. We consider two cases.
  Either $M_{k0}=(\bot,=)$. We let $v'$ be the valuation
  s.t. $v'(x)=\bot$ and for all $y\neq x$: $v'(y)=v(y)$. Clearly
  $v'\in\denote{M}$ since $M$ is non-empty and in normal form. Hence,
  $v\in \release_x{\denote{M}}$. Otherwise $M_{k0}=(m,\prec)$ with
  $m\in\posreal\cup\{\infty\}$, since we have already ruled out the
  case $M_{k0}=(?,=)$. We let $v'$ be a valuation that is a solution
  of the following set of inequalities if $x$ is an history clock:
  \begin{displaymath}
    \begin{array}{rcll}
      v'(y) &=&v(y) & \textrm{ for all $y\neq x$}\\
      -m_{0k}\prec_{0k}v'(x)&\prec_{k0}& m_{k0} \\
      -m_{jk}\prec_{jk} v'(x) - v'(x_j)&\prec_{kj}& m_{kj}& \textrm{ for all $x_j\in(\hclocks{\Sigma}\cap C)\setminus\{x\}$}\\
      -m_{jk}\prec_{jk} v'(x) + v'(x_j)&\prec_{kj}& m_{kj}& \textrm{ for all $x_j\in(\pclocks{\Sigma}\cap C)\setminus\{x\}$}\\
    \end{array}
  \end{displaymath}
  or a solution of the following set of inequalities if $x$ is a
  prophecy clock:
  \begin{displaymath}
    \begin{array}{rcll}
      v'(y) &=&v(y) & \textrm{ for all $y\neq x$}\\
      -m_{k0}\prec_{k0}v'(x)&\prec_{0k}& m_{0k} \\
      -m_{kj}\prec_{kj} v'(x) - v'(x_j)&\prec_{jk}& m_{jk}& \textrm{ for all $x_j\in(\pclocks{\Sigma}\cap C)\setminus\{x\}$}\\
      -m_{jk}\prec_{jk} v'(x) + v'(x_j)&\prec_{kj}& m_{kj}& \textrm{ for all $x_j\in(\hclocks{\Sigma}\cap C)\setminus\{x\}$}\\
    \end{array}
  \end{displaymath}
  assuming as usual that $\bot+c=c+\bot=\bot-c=c-\bot=\bot$.

  Since $M$ is in normal form, such a $v'$ exists (otherwise, some of
  the constraints could be strengthened without modifying the zone,
  and $M$ is not in normal form), and it is in
  $\denote{M}$. Hence $v$ is in $\release_x(\denote{M})$. We conclude
  that $\denote{\release_x(M)}\subseteq\release_x(\denote{M})$.
\item The proof stems from the fact that
  $\denote{M^1}\subseteq\denote{M^2}$ {\bf iff}
  $\denote{M^1}\cap\denote{M^2}=\denote{M^1}$ {\bf iff}
  $\denote{M^1\cap M^2}=\denote{M^1}$ {\bf iff},
  $\min(M^1_{i,j},M^2_{i,j})=M_{i,j}$ for all $0\leq i,j\leq n$ (by
  Proposition~\ref{prop:construction-EDBMs}).
\end{enumerate}
\end{proof}

\paragraph{Forward and backward analysis}
We present now the forward and backward analysis algorithms adapted to
\ecta. From now on, we consider an \ecta
$A=\tuple{Q,q_i,\Sigma,\delta,\alpha}$.  We also let $\post{(q,v)} =
\{(q',v')\mid\exists t, a : (q,v)\xrightarrow{t,a} (q',v')\}$ and
$\pre{(q,v)}=\{(q',v')\mid\exists t, a : (q',v')\xrightarrow{t,a}
(q,v)\}$ and we extend those operators to sets of states in the
natural way.  Moreover, given a set of valuations $Z$ and a location
$q$, we abuse notations and denote by $(q,Z)$ the set $\{(q,v)\mid
v\in Z\}$. Also, we let
$\poststar{(q,Z)}=\bigcup_{n\in\mathbb{N}}\posti{(q,Z)}{n}$ and
$\prestar{(q,Z)}=\bigcup_{n\in\mathbb{N}}\prei{(q,Z)}{n}$, where
$\posti{(q,Z)}{0}=(q,Z)$ and
$\posti{(q,Z)}{n}=\post{\posti{(q,Z)}{n-1}}$, and similarly for
$\prei{(q,Z)}{n}$. The \postop and \preop operators are sufficient to
solve language emptiness for \ecta:

\begin{lemma}[adapted from \cite{297329}, Lemma
  1]\label{lem:zone-analysis}
  Let $A=\tuple{Q,q_i,\Sigma,\delta,\alpha}$ be an \ecta, let
  $I=\{(q_i,v)\mid v\textrm{ is initial}\}$, and
  let $\overline{\alpha}=\{(q,v)\mid q\in\alpha \textrm{ and
  }v\textrm{ is final}\}$. Then:
  
  \centerline{$\poststar{I}\cap \overline{\alpha}\neq\emptyset$ iff
    $\prestar{\overline{\alpha}}\cap I\neq\emptyset$ iff
    $L(A)\neq\emptyset$.}
\end{lemma}

Let us show how to compute these operators on \ezone s. Given a
location $q$, an \ezone{} $Z$ on $\clocks{\Sigma}$, and an edge
$e=(q,a,\psi,q')\in\delta$, we let:
{\small
\begin{eqnarray*}
  \poste{e}{(q_1,Z)}&=&%
  \begin{cases}
    \left(q',\Bigl(\release_{\hclock{a}}\bigl(\release_{\pclock{a}}(\future{Z}\cap(\pclock{a}=0))\cap \psi\bigr)\Bigr)\cap (\hclock{a}=0)\right)&\textrm{if }q_1= q\\
    \emptyset&\textrm{otherwise}\\
  \end{cases}\\
  \pree{e}{(q_1,Z)}&=&%
  \begin{cases}
    \left(q,\past{\bigl(\release_{\pclock{a}}(\release_{\hclock{a}}(Z\cap(\hclock{a}=0))\cap \psi)\bigr)\cap (\pclock{a}=0)}\right)&\textrm{if }q_1= q'\\
    \emptyset&\textrm{otherwise}\\
  \end{cases}
\end{eqnarray*}
} Then, it is easy to check that
$\post{(q,Z)}=\cup_{e\in\delta}\poste{e}{(q,Z)}$ and
$\pre{(q,Z)}=\cup_{e\in\delta}\pree{e}{(q,Z)}$. With the algorithms on
EDBMs presented above, these definitions can be used to compute the
\textsf{Pre} and \textsf{Post} of zones using their EDBM
encodings. Remark that \textsf{Pre} and \textsf{Post} return
\emph{sets of \ezone s} as these are not closed under union.

\begin{algorithm}[t]
\caption{The forward and backward algorithms \label{algo:forward-analysis}}
\forwex\Begin{
  Let \texttt{Visited} = $\emptyset$ ; Let \texttt{Wait} = $\{(q_i,Z_0)\}$ \;
  \While{$\texttt{Wait}\neq\emptyset$}{
    Get and remove $(q,Z)$ from \texttt{Wait} \;
    \lIf{$q\in\alpha$ and $Z\subseteq Z_f$}{
      \Return{\texttt{Yes}} \;
    }
    \If{there is no $(q,Z')\in\texttt{Visited}$ s.t. $Z\subseteq Z'$}{
      $\texttt{Visited}:=\texttt{Visited} \cup \{(q,Z)\}$ ;
      $\texttt{Wait} := \texttt{Wait}\cup \post{(q,Z)}$ \;
    }
  }
  \Return{\texttt{No}} \;
}
\bigskip

\backex\Begin{
  Let \texttt{Visited} = $\emptyset$ ; Let \texttt{Wait} = $\{(q,Z_f)\mid q\in\alpha\}$ \;
  \While{$\texttt{Wait}\neq\emptyset$}{
    Get and remove $(q,Z)$ from \texttt{Wait} \;
    \lIf{$q=q_i$ and $Z\subseteq Z_0$}{
      \Return{\texttt{Yes}} \;
    }
    \If{there is no $(q,Z')\in\texttt{Visited}$ s.t. $Z\subseteq Z'$ \label{ln:if-back}}{
      $\texttt{Visited}:=\texttt{Visited} \cup \{(q,Z)\}$ ;
      $\texttt{Wait} := \texttt{Wait}\cup \pre{(q,Z)}$ \;
    }
  }
  \Return{\texttt{No}} \;
}
\end{algorithm}

\begin{figure}
\centerline{\includegraphics{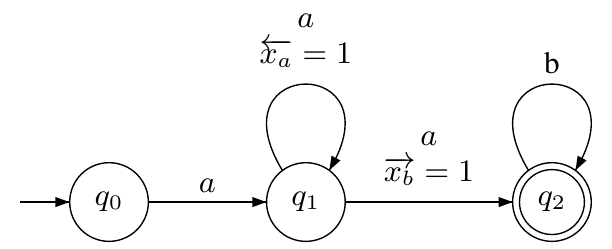}}
\caption{An \ecta for which backward analysis does not
  terminate.}\label{fig:backward-analysis}
\end{figure}

Let us now consider the \forwex and \backex algorithms to test for
language emptiness of \ecta, shown in
Algorithm~\ref{algo:forward-analysis}. In these two algorithms $Z_0$
denotes the zone $\bigwedge_{x\in\hclocks{\Sigma}}x=\bot$ containing
all the possible initial valuations and $Z_f$ denotes the zone
$\bigwedge_{x\in\pclocks{\Sigma}}x=\bot$ representing all the possible
final valuations. By Lemma~\ref{lem:zone-analysis}, it is clear that
\forwex and \backex are correct when they terminate. Unfortunately,
Fig.~\ref{fig:backward-analysis} shows an \ecta on which the
backward algorithm does not terminate. Since history and prophecy
clocks are symmetrical, this example can be adapted to define an \ecta
on which the forward algorithm does not terminate either. Remark that
in the case of timed automata, the forward analysis is not guaranteed
to terminate, whereas the backward analysis \emph{always terminates}
(the proof relies on a bisimulation argument)~\cite{DBLP:conf/cav/Alur99}.

\begin{proposition}\label{prop:forw-back-do-not-terminate}
  Neither \forwex\ nor \backex\ terminate in general.
\end{proposition}
\begin{proof}
  We give the proof for \backex, a similar proof for \forwex\ can then
  be deduced by symmetry.  Consider the \ecta in
  Fig.~\ref{fig:backward-analysis}. Running the backward analysis
  algorithm from $(q_2, Z_f)$, we obtain, after selecting the
  transition $e=(q_2,b,\texttt{true}, q_2)$, the zone
  $Z_1=\pclock{a}=\bot\wedge \hclock{b}=0$. Then, the transition
  $e'=(q_1,a,\pclock{b}=1, q_2)$ is back-firable and we attain the
  zone $Z_2=\pclock{a}\geq 0\wedge \pclock{b}\geq 1\wedge
  \pclock{b}-\pclock{a}=1$. At this point the transition
  $e''=(q_1,a,\hclock{a}=1,q_1)$ is back-firable, which leads to the
  zone $Z_3=\pclock{b}\geq 1\wedge \pclock{a}\geq 0\wedge 0\leq
  \hclock{a}\leq 1 \wedge \pclock{b} - \pclock{a}\geq 1 \wedge
  \pclock{b}+\hclock{a}\geq 2$. The back-firing of the $e''$
  transition can be repeated, and, by induction, after $n$ iterations
  of the loop, the algorithm reaches the zone $Z^n=\pclock{b}\geq
  n\wedge \pclock{a}\geq 0\wedge 0\leq \hclock{a}\leq 1 \wedge
  \pclock{b} - \pclock{a}\geq n \wedge \pclock{b}+\hclock{a}\geq
  n+1$. Thus, the condition of the \textbf{if} in
  line~\ref{ln:if-back} is always fulfilled, and the algorithm visits
  an infinite number of zones, without reaching $q_0$.  
\end{proof}

\begin{figure}
  \begin{center}
    \begin{tabular}{lcclc}
      (a)&\includegraphics[scale=.5]{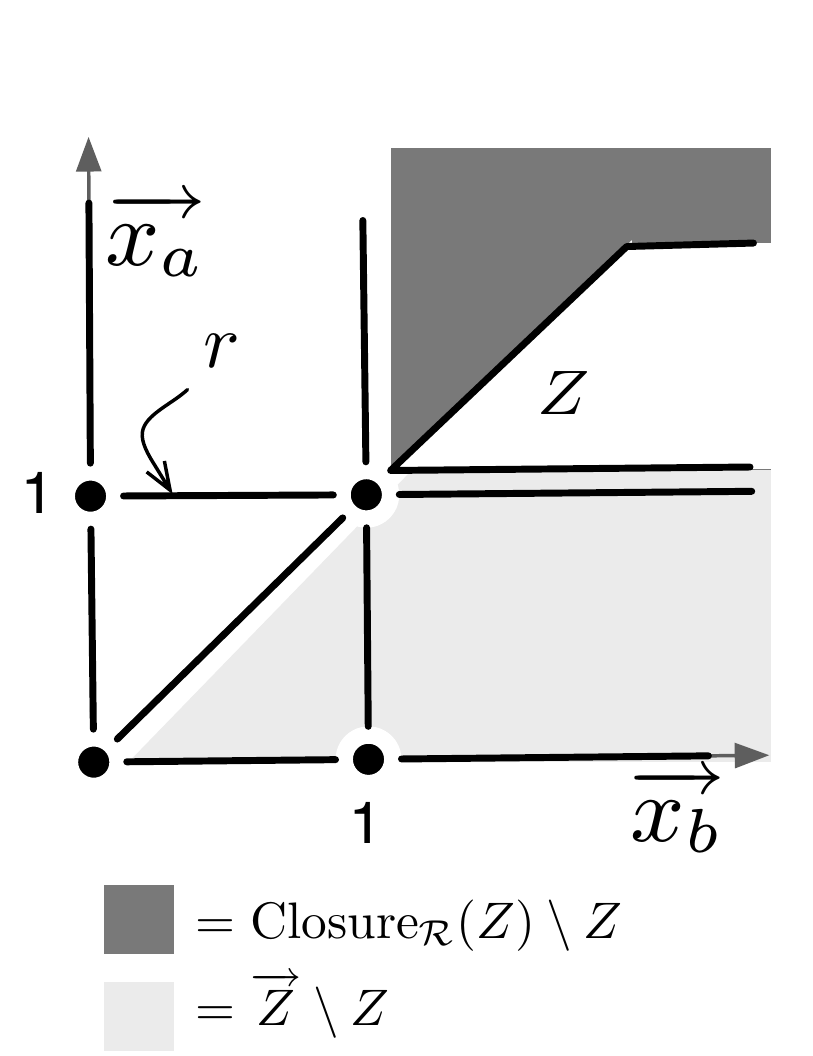}&\hspace*{.1cm}
      (b)&\includegraphics[scale=.5]{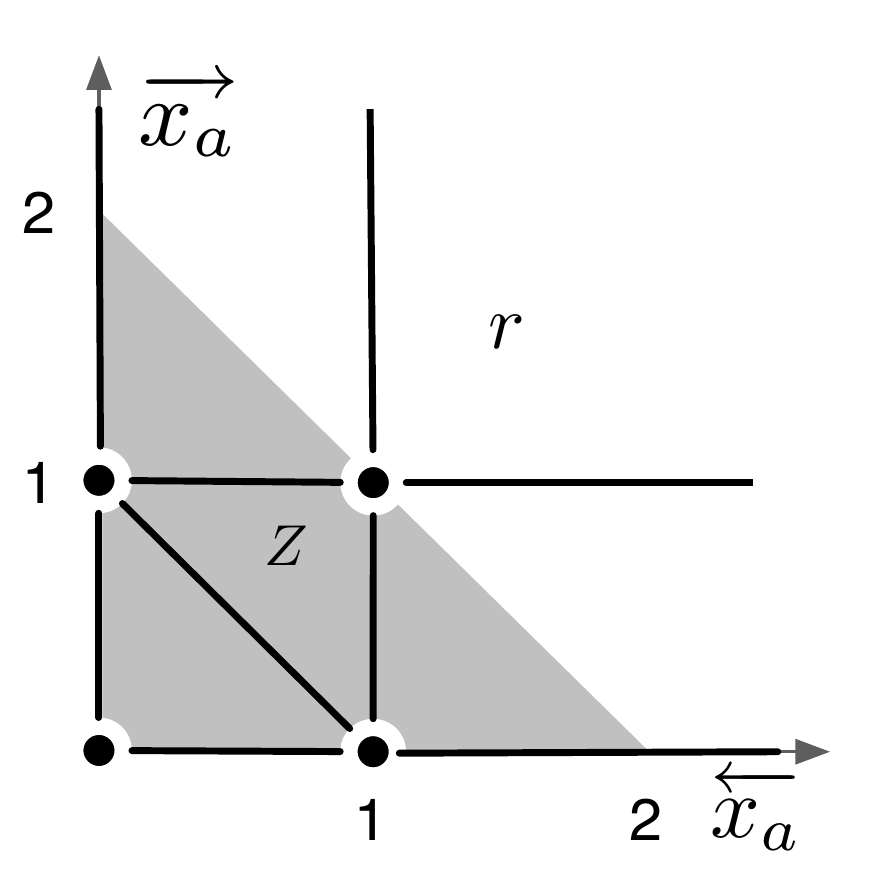}\\
    \end{tabular}
  \end{center}

\caption{ Examples for $\closure{R}$ and
 $\appro{k}$} \label{fig:closure}
\end{figure}

\section{Future work: widening operators}

As said earlier, the zone-based \emph{forward} analysis algorithm does
not terminate either in the case of timed automata. To recover
termination, \emph{widening operators} have been defined.
The most popular widening operator is the so-called $k$-approximation
on zones \cite{DBLP:conf/tacas/T98}. Roughly speaking, it is
defined as follows: in the definition of the zone, replace any
constraint of the form $x_i\prec c$ or $x_i-x_j\prec c$, by
respectively $x_i < \infty$ and $x_i-x_j < \infty$ if and only if $c>
k$, and replace any constraint of the form $c \prec x_i$ or $c \prec
x_i-x_j$, by respectively $ k < x_i$ and $k < x_i-x_j$, if and only if
$c > k$. Such an operator can be easily computed on DBMs, and is a
standard operation implemented in several tools such as as UppAal
\cite{DBLP:conf/qest/BehrmannDLHPYH06} for more more than 15
years. Nevertheless, this operator has been widely discussed in the
recent literature since Bouyer has pointed out several flaws in the
proposed proofs of soundness \cite{bouyer-fmsd-2004}. Actually, the
$k$-approximation is \emph{sound} when the timed automaton contains
\emph{no diagonal constraints}. Unfortunately, $k$-approximation is
\emph{not sound} when the timed automaton contains diagonal
constraints, and \emph{no sound widening operator exists} in this
case.

In \cite{bouyer-fmsd-2004}, Bouyer identifies some subclasses of timed
automata for which the widening operator is provably correct. The idea
of the proof relies mainly on the definition of another widening
operator, called the \emph{closure by regions}, which is shown to be
\emph{sound}. The closure by regions of a zone $Z$, with respect to a
set of regions $\cR$ is defined as the smallest set of regions from
$\cR$ that have a non-empty intersection with $Z$, i.e.
$\closure{\cR}(Z)=\{r\in\cR\mid Z\cap r\neq\emptyset\}$.  Then, the
proof concludes by showing that $\appro{k}(Z)$ is sound for some
values of $k$ (that are proved to exist) s.t.
\begin{equation}\label{eq:correctness-approx}
Z \subseteq \appro{k}(Z)\subseteq \closure{\cR}(Z).
\end{equation}

In the perspective of bringing \ecta from theory to implementation,
\emph{provably correct} widening operators are necessary, since
neither the forward nor the backward algorithm terminate in
general. We plan to adapt the $k$-approximation to \ecta, and we
believe that we can follow the general idea of the proof in
\cite{bouyer-fmsd-2004}.  However, the proof techniques will not be
applicable in a straightforward way, for several reasons. First, the
proof of \cite{bouyer-fmsd-2004} relies on the following property,
which holds in the case of timed automata: for all zone $Z$ and all
location $q$:
$\post{(q,\closure{\cR}(Z))}\subseteq\closure{\cR}(\post{(q,Z)})$.
Unfortunately this is not the case in general with \ecta. Indeed,
consider the zone $Z$ and the region $r$ in Fig.~\ref{fig:closure}
(a). Clearly, $r$ is included in $\future{\closure{\cR}(Z)}$ but $r$
is not included in $\closure{\cR}(\future{Z})$ (recall that prophecy
clocks decrease with time elapsing). Moreover, the definition of the
$k$ approximation will need to be adapted to the case of
\ecta. Indeed, the second inclusion in~(\ref{eq:correctness-approx})
does not hold when using the $k$-approximation defined for timed
automata, which merely replaces all constants $>k$ by $\infty$ in the
constraints of the zone. Indeed, consider the \ezone{} $Z$ defined by
$\hclock{a}+\pclock{a}\leq 2$ in Fig.~\ref{fig:closure}~(b), together
with the set of regions $\cR=\Reg{\clocks{\{a\}}}{1}$. Clearly, with
such a definition, the constraint $\hclock{a}+\pclock{a}\leq 2$ would
be replaced by $\hclock{a}+\pclock{a}<\infty$, which yields an
approximation that intersects with $r$, and is thus not contained in
$\closure{\cR}(Z)$. We keep open for future works the definition of a
provably correct adaptation of the $k$-approximation for \ecta.

\bibliography{Verif}

\begin{thebibliography}{10}

\bibitem{DBLP:conf/cav/Alur99}
R.~Alur.
\newblock Timed automata.
\newblock In {\em Proceedings of CAV'99}, volume 1633 of {\em Lecture Notes in
  Computer Science}, pages 8--22. Springer, 1999.

\bibitem{AD94}
R.~Alur and D.~Dill.
\newblock A {T}heory of {T}imed {A}utomata.
\newblock {\em Theoretical Computer Science}, 126(2):183--236, 1994.

\bibitem{297329}
R.~Alur, L.~Fix, and T.~A. Henzinger.
\newblock Event-clock automata: a determinizable class of timed automata.
\newblock {\em Theoretical Computer Science}, 211(1-2):253--273, 1999.

\bibitem{DBLP:conf/qest/BehrmannDLHPYH06}
G.~Behrmann, A.~David, K.~G. Larsen, J.~H{\aa}kansson, P.~Pettersson, W.~Yi,
  and M.~Hendriks.
\newblock Uppaal 4.0.
\newblock In {\em Proceedings of QEST'06}, pages 125--126. IEEE Computer
  Society, 2006.

\bibitem{Bellman57}
R.~Bellman.
\newblock {\em Dynamic Programming}.
\newblock Princeton university press, 1957.

\bibitem{bouyer-fmsd-2004}
P.~Bouyer.
\newblock Forward analysis of updatable timed automata.
\newblock {\em Formal Methods in System Design}, 24(3):281--320, May 2004.

\bibitem{DBLP:conf/cav/BozgaDMOTY98}
M.~Bozga, C.~Daws, O.~Maler, A.~Olivero, S.~Tripakis, and S.~Yovine.
\newblock Kronos: A model-checking tool for real-time systems.
\newblock In {\em Proceedings of CAV'98}, volume 1427 of {\em Lecture Notes in
  Computer Science}, pages 546--550. Springer, 1998.

\bibitem{DBLP:conf/tacas/T98}
C.~Daws and S.~Tripakis.
\newblock Model checking of real-time reachability properties using
  abstractions.
\newblock In B.~Steffen, editor, {\em Proceedings of TACAS'98}, volume 1384 of
  {\em Lecture Notes in Computer Science}, pages 313--329. Springer, 1998.

\bibitem{GGRS}
B.~Di~Giampaolo, G.~Geeraerts, J.~Raskin, and N.~Sznajder.
\newblock Safraless procedures for timed specifications.
\newblock In {\em Proceedings of FORMATS'10}, volume 6246 of {\em Lecture Notes
  in Computer Science}, pages 2--22. Springer, 2010.

\bibitem{Dill89}
D.~L. Dill.
\newblock Timing assumptions and verification of finite-state concurrent
  systems.
\newblock In {\em Proceedings of Automatic Verification Methods for Finite
  State Systems}, volume 407 of {\em Lecture Notes in Computer Science}, pages
  197--212. Springer, 1989.

\bibitem{Dima:1999:KTE:647899.740957}
C.~Dima.
\newblock Kleene theorems for event-clock automata.
\newblock In {\em Proceedings of FCT'99}, volume 1684 of {\em Lecture Notes in
  Computer Science}, pages 215--225. Springer, 1999.

\bibitem{DBLP:conf/formats/DSouzaT04}
D.~D'Souza and N.~Tabareau.
\newblock On timed automata with input-determined guards.
\newblock In {\em Proccedings of FORMATS/FTRTFT'04}, volume 3253 of {\em
  Lecture Notes in Computer Science}, pages 68--83, 2004.

\bibitem{339240}
J.-F. Raskin and P.-Y. Schobbens.
\newblock The logic of event clocks: decidability, complexity and
  expressiveness.
\newblock {\em Automatica}, 34(3):247--282, 1998.

\bibitem{Sor01}
M.~Sorea.
\newblock Tempo: {A} model-checker for event-recording automata.
\newblock In {\em Proceedings of RT-TOOLS'01}, Aalborg, Denmark, August 2001.

\bibitem{Tang:2009:EVP:1506688.1506740}
N.~Tang and M.~Ogawa.
\newblock Event-clock visibly pushdown automata.
\newblock In {\em Proceedings of SOFSEM'09}, volume 5404 of {\em Lecture Notes
  in Computer Science}, pages 558--569. Springer, 2009.

\end{thebibliography}

\end{document}

